\documentclass[pra,aps,superscriptaddress,twocolumn,letter,nopacs,nofootinbib,longbibliography,notitlepage]{revtex4-1}

\usepackage{graphicx,color,amsmath,amsfonts,enumerate,amsthm,amssymb,mathtools,enumitem,thmtools,hyperref,subfigure,mathdots,enumitem,centernot,bm,soul,bbm}
\usepackage[capitalise, noabbrev]{cleveref}
\hypersetup{colorlinks=true,linkcolor=blue,citecolor=blue,urlcolor=blue}

\def\>{\rangle}
\def\<{\langle}

\newcommand{\ket}[1]{| {#1} \rangle}
\newcommand{\abs}[1]{\left| {#1} \right|} 
\newcommand{\ketbra}[2]{\ensuremath{\left|#1\right\rangle\!\!\left\langle#2\right|}}
\newcommand{\braket}[2]{\ensuremath{\!\!\left\langle#1|#2\right\rangle}\!\!}
\newcommand{\matrixel}[3]{\ensuremath{\left\langle #1 \vphantom{#2#3} \right| #2 \left| #3 \vphantom{#1#2} \right\rangle}}

\renewcommand{\v}[1]{\ensuremath{\boldsymbol #1}}

\definecolor{ppblue}{RGB}{46,117,182}
\definecolor{ppred}{RGB}{197, 90, 17}

\newcommand{\bmtilde}[1]{\bm{\tilde{#1}}}

\DeclareMathOperator{\Var}{Var}

\newcommand{\ev}[1]{\stackrel{\text{ev.}}{#1}}
\theoremstyle{plain}
\newtheorem{thm}{Theorem}
\newtheorem{lem}[thm]{Lemma}
\newtheorem{prop}[thm]{Proposition}

\newcounter{def}
\theoremstyle{definition}
\newtheorem{defn}[def]{Definition}
\newtheorem{rmk}[thm]{Remark}

\usepackage{tikz,pgfplots}
\begin{document}

\title{Moderate deviation analysis of majorisation-based resource interconversion}
\author{Christopher T.~Chubb}
\email{christopher.chubb@sydney.edu.au}
\affiliation{Centre for Engineered Quantum Systems, School of Physics, University of Sydney, Sydney NSW 2006, Australia.}
\author{Marco Tomamichel}
\affiliation{Centre for Quantum Software and Information, School of Software, University of Technology Sydney, Sydney NSW 2007, Australia.}
\author{Kamil Korzekwa}
\affiliation{Centre for Engineered Quantum Systems, School of Physics, University of Sydney, Sydney NSW 2006, Australia.}

\begin{abstract}
We consider the problem of interconverting a finite amount of resources within all theories whose single-shot transformation rules are based on a majorisation relation, e.g.\ the resource theories of entanglement and coherence (for pure state transformations), as well as thermodynamics (for energy-incoherent transformations). When only finite resources are available we expect to see a non-trivial trade-off between the rate $r_n$ at which $n$ copies of a resource state $\rho$ can be transformed into $nr_n$ copies of another resource state $\sigma$, and the error level $\epsilon_n$ of the interconversion process, as a function of $n$.
In this work we derive the optimal trade-off in the so-called moderate deviation regime, where the rate of interconversion $r_n$ approaches its optimum in the asymptotic limit of unbounded resources ($n\to\infty$), while the error $\epsilon_n$ vanishes in the same limit. We find that the moderate deviation analysis exhibits a resonance behaviour which implies that certain pairs of resource states can be interconverted at the asymptotically optimal rate with negligible error, even in the finite $n$ regime.
\end{abstract}

\maketitle

\section{Introduction}
\label{sec:intro}

In principle, while processing quantum information, any initial state can be transformed into any final state. One could thus conclude that all quantum states are equally valuable or resourceful. In reality, however, some transformations are harder to implement than others, which results in a partial ordering of the set of quantum states, with the hardest to prepare at the top, and easiest at the bottom. Such a resource hierarchy arises naturally when we face any kind of restrictions: from the locality constraint, through experimental difficulties in preparing particular superpositions, to fundamental constraints induced by physical laws like energy conservation. The mathematical framework developed to study possible state transformations under such restrictions is known under the collective name of resource theories~\cite{coecke2016mathematical,chitambar2018quantum}.

Inspired by classical information theory, the early resource-theoretic works considered optimal conversion rates between different resource states in the asymptotic regime, i.e., the limit of processing infinitely many copies of a given state. This led to the discovery of asymptotic resource measures, which provided operational meaning to quantities such as entropy of entanglement~\cite{bennett1996concentrating} and non-equilibrium free energy~\cite{brandao2013resource}. Namely, a given transformation becomes asymptotically possible if and only if the corresponding asymptotic resource measure is non-increasing, which allows one to reversibly interconvert between all resource states.

On the other hand, almost simultaneously to the asymptotic studies, the single-shot regime was investigated, where one aims at deciding whether it is possible to convert a single copy of an initial state into the final state. Here, probably the most famous contributions are the Nielsen's theorem~\cite{nielsen1999conditions} within the resource theory of entanglement, and more recently the family of second laws for the resource theory of thermodynamics~\cite{brandao2015second}. In general, in the single-shot regime simple asymptotic transformation rules get replaced by more complex sets of conditions, which also give rise to irreversible transformation.

In this paper we focus on the interconversion process in the intermediate regime, when the number of processed resource states is large, but finite. This way we aim at keeping the simplicity of the asymptotic analysis, but also at preserving the irreversible nature of single-shot regime. The first steps in this direction were recently made in Refs.~\cite{kumagai2017second}~and~\cite{chubb2017beyond} for the resource theories of entanglement and thermodynamics, where the corrections to asymptotic conversion rates were found in the scenario with a constant transformation error (i.e., in the small deviation regime~\cite{strassen62}). Here, we present a moderate deviation analysis~\cite{altug14} (see also~\cite{Chubb2017,cheng17} for applications in the quantum domain) of the interconversion problem within a unified framework that includes all resource theories for which the single-shot transformation rules can be expressed via majorisation or thermo-majorisation. This way we find finite-size corrections to conversion rates in resource theories of entanglement~\cite{horodecki2009quantum}, coherence~\cite{baumgratz2014quantifying} and thermodynamics~\cite{horodecki2013fundamental}, in the regime where the transformation error, measured by either infidelity or total variation distance, asymptotically vanishes. 

Our results can be directly applied to the study of important problems such as entanglement distillation~\cite{kwiat2001experimental} or coherence dilution~\cite{winter2016operational}, but also allow one for a rigorous analysis of the irreversibility arising when finite-size resources are interconverted. Most intriguingly, we find that if a pair of states satisfies a particular resonance condition, one can achieve lossless interconversion, i.e., transformation that is arbitrarily close to reversible even for finite $n$. In the accompanying paper~\cite{korzekwa2018avoiding} we discuss how this effect can be employed to avoid irreversibility, which directly affects, e.g., the performance of heat engines working with finite-size working bodies~\cite{tajima2017finite}.

This paper is structured in the following way. First, in Sec.~\ref{sec:setting}, we set the scene by introducing necessary tools and concepts. Next, in Sec.~\ref{sec:result} we state our main result concerning moderate deviation corrections to the asymptotic interconversion rates for majorisation-based resource theories. We then proceed to Sec.~\ref{sec:tail_bounds} that contains auxiliary technical results concerning tail bounds, which are used in the formal proof that can be found in Sec.~\ref{sec:proof}. Finally, we provide conclusions and outlook in Sec.~\ref{sec:outlook}.

\section{Setting the scene}
\label{sec:setting}

\subsection{Resource theories in different regimes}

Every quantum resource theory~\cite{coecke2016mathematical,chitambar2018quantum} is defined by a set of quantum operations that are considered free, and a set of restrictions that make other operations impossible without an additional cost. Such restrictions may arise from practical difficulties, e.g., when preparing a system in a superposition of particular states is experimentally challenging, but may also be of fundamental nature, as with the laws of thermodynamics constraining possible transformations to preserve energy and increase entropy. A resource is then defined as a quantum system that allows one to lift a given restriction. Typical examples of resources include an excited pure state that acts as a work storage, and thus can be used to reduce the entropy of another system (overcoming thermodynamic constraints~\cite{horodecki2013fundamental}); an entangled Bell pair, which can be used to teleport a quantum state (overcoming locality constraints~\cite{horodecki2009quantum}); or a system in the superposition of energy eigenstates, which can be used as a reference frame for time (overcoming symmetry constraints~\cite{bartlett2007reference}).

Once the restrictions and the corresponding resources are defined, the central question concerns resource interconversion, i.e., what final states can be obtained from the initial one. This problem so far was mainly approached in either the \emph{single-shot regime}~\cite{gour2017quantum_resource}, or in an idealised \emph{asymptotic limit}~\cite{brandao2015reversible}. The first approach, due to its generality and the corresponding complexity of the answer, provides only a limited insight into the nature of different resource states. The second one provides an elegant and simple answer in the form of optimal \emph{conversion rate}, which tells us how many copies of the final state can be obtained per one copy of the initial state, if one assumes access to an infinite source of initial states. From a practical point of view, however, such an assumption is unjustified, as most quantum resources will be available only in small amounts in the foreseeable future. More fundamentally, finite-size effects may be of interest themselves, as it is the case within quantum thermodynamics~\cite{goold2016role}, where one aims at accurate description of heat and work processes involving small number of particles.

Very recently the first steps have been made to study the intermediate regime, where one focuses on the interconversion of large but finite number $n$ of resource states. First, in Ref.~\cite{kumagai2017second} the authors focused on transformations within the resource theory of entanglement. Their results were then generalised and adapted to the studies of the interconversion process in the resource theory of thermodynamics by the present authors~\cite{chubb2017beyond}. In both these works the second-order correction to the asymptotic rate was found in the so-called \emph{small deviation} regime~\cite{strassen62}, where the conversion rate approaches the asymptotic one for $n\rightarrow\infty$, but the transformation is realised with a constant error. In the current work we solve the issue of constant error by deriving corrections to the asymptotic rate in the \emph{moderate deviation} regime~\cite{altug14}, where the correction term still vanishes as $n\to\infty$, but also the transformation is asymptotically error-free. For the completeness of discussion, we also note that the interconversion problem may be studied in the \emph{large deviation} regime~\cite{dembo98}, where the error is exponentially vanishing for the price of the constant gap between the realised conversion rate and the asymptotic one. In Table~\ref{table:regimes} we collect references to central results concerning state interconversion within the investigated resource theories in various regimes.

\begin{table}
	\renewcommand{\arraystretch}{1.2}
	\begin{tabular}{c|c|c|c|}	
		\cline{2-4}
		& ~~~Ent.~~~ &~~~ Coh.~~~ &~Thermo.~~\\	
		\cline{1-4}
		\multicolumn{1}{|c|}{$n=1$, $\epsilon=0$}	& \cite{nielsen1999conditions} & \cite{du2015conditions}& \cite{horodecki2013fundamental}\\
		\cline{1-4}
		\multicolumn{1}{|c|}{$n\to\infty$, $\epsilon\to 0$}		& \multicolumn{2}{|c|}{\cite{bennett1996concentrating}} & \cite{brandao2013resource}\\
		\cline{1-4}
		\multicolumn{1}{|c|}{$n<\infty$, $\epsilon>0$}		& \multicolumn{2}{|c|}{\cite{kumagai2017second}} & \cite{chubb2017beyond}\\
		\cline{1-4}
		\multicolumn{1}{|c|}{$n<\infty$, $\epsilon\to 0$}	&	\multicolumn{3}{|c|}{This work}\\	
		\cline{1-4}
	\end{tabular}
	\caption{\emph{Interconversion in various regimes.} Exposition of works on state interconversion within resource theories of ent(anglement), coh(erence) and thermo(dynamics) in single-shot regime ($n=1,\epsilon=0$), asymptotic limit \mbox{($n\to\infty,\epsilon\to 0$)}, small deviation regime ($n<\infty,\epsilon>0$) and moderate deviation regime ($n<\infty,\epsilon\to 0$).}
	\label{table:regimes}
\end{table}

\subsection{Exact single-shot interconversion}

Irrespective of the investigated regime, the first step is to find single-shot interconversion rules, which form the basis of further analysis. In this work we study the interconversion problem within all majorisation-based resource theories, i.e., when conditions for single-shot transformations can be expressed as majorisation partial order~\cite{marshall1979inequalities}, or a variant known as thermo-majorisation~\cite{ruch1978mixing,horodecki2013fundamental}. Within such theories, each resource state can be represented by a probability distribution, and the conversion process is possible when the distribution representing the initial state majorises (or is majorised) by the distribution representing the final state, with majorisation $\succ$ defined by
\begin{equation}
	\v{a}\succ\v{b}\quad\Longleftrightarrow\quad \forall j:~\sum_{i=1}^j a_i^\downarrow\geq \sum_{i=1}^j b_i^\downarrow,
\end{equation}
where $\v{a}^\downarrow$ denotes the vector $\v{a}$ in a decreasing order.

Three prominent examples of majorisation-based resource theories include the resource theories of entanglement, coherence and thermodynamics. These are defined via the relevant sets of free operations and free states: Local Operations and Classical Communication (LOCC) and separable states in entanglement theory~\cite{horodecki2009quantum}; Incoherent Operations and incoherent states in coherence theory~\cite{baumgratz2014quantifying}; Thermal Operations and the thermal equilibrium state~$\gamma$ in the resource theory of thermodynamics (with respect to a fixed background temperature $T=1/\beta$)~\cite{janzing2000thermodynamic}. As mentioned above, within each of these theories there exists a representation of initial and target quantum states, $\rho$ and $\sigma$, as probability distributions $\v{p}$ and $\v{q}$. For entanglement theory, given initial and target pure bipartite states, $\rho=\ketbra{\Psi}{\Psi}$ and $\sigma=\ketbra{\Phi}{\Phi}$, with the Schmidt decomposition given by
\begin{equation}
	\ket{\Psi}=\sum_i a_i \ket{\psi_i\psi_i},\quad \ket{\Phi}=\sum_i b_i \ket{\phi_i\phi_i},
\end{equation}
we can represent them via probability distributions
\begin{equation}
	\label{eq:ent:representation}
	p_i=|a_i|^2,\quad q_i=|b_i|^2.
\end{equation}
For coherence theory, with respect to a fixed basis $\{\ket{i}\}$, one can represent pure initial and target states, \mbox{$\rho=\ketbra{\psi}{\psi}$} and $\sigma=\ketbra{\phi}{\phi}$, using 
\begin{equation}
	\label{eq:coh:representation}
	p_i=|\braket{i}{\psi}|^2,\quad q_i=|\braket{i}{\phi}|^2.
\end{equation}
Finally, in the resource theory of thermodynamics, the initial and target energy-incoherent mixed states $\rho$ and $\sigma$ can be represented by
\begin{equation}
	\label{eq:thermo:representation}
	p_i=\matrixel{E_i}{\rho}{E_i},\quad q_i=\matrixel{E_i}{\sigma}{E_i},
\end{equation}
where $\{\ket{E_i}\}$ denotes the energy eigenbasis of the system. We will denote distributions representing free states by $\v{f}$. In entanglement and coherence theories these are represented by sharp probability distributions $\v{s}$ with a single non-zero entry; whereas in the thermodynamic case $\v{f}$ is given by a thermal Gibbs distribution $\v{\gamma}$ with $\gamma_i\propto\exp(-\beta E_i)$. 

The celebrated Nielsen's theorem~\cite{nielsen1999conditions} (for entanglement) and the recent result of Ref.~\cite{du2015conditions} (for coherence) state that the initial state represented by $\v{p}$ can be transformed into the target state represented by $\v{q}$ if and only if $\v{p}\prec\v{q}$. Similarly, in Ref.~\cite{horodecki2013fundamental}, it was found that a thermodynamic transformations between states represented by $\v{p}$ and $\v{q}$ is possible if and only if $\hat{\v{p}}\succ\hat{\v{q}}$, where $\hat{\v{a}}$ can be obtained from $\v{a}$ via a straightforward application of an \emph{embedding map} $\Gamma^\beta$~\cite{brandao2015second,lostaglio2018thermodynamic}. For the sake of our analysis, it is only crucial to note that $\Gamma^\beta$ maps $d$-dimensional distributions to $\hat{d}$-dimensional ones with $\hat{d}\geq d$; and that an embedded version of the free thermal distribution is given by a maximally mixed distribution on a larger subspace, i.e., $\hat{\v{\gamma}}=\v{\eta}$ with $\v{\eta}=[1/\hat{d},\dots,1/\hat{d}]$.

\subsection{Approximate multi-copy interconversion}
\label{subsec:multicopy}

When considering transformations between many copies of initial and target states, represented by $\v{p}^{\otimes n}$ and $\v{q}^{\otimes m}$, we need to make sure that the dimensionality of the input and output spaces match. Since one can always append any number of free states $\v{f}$ to both the initial and target states, we introduce total initial and target distributions,
\begin{equation}
	\v{P}^{n,m}:=\v{p}^{\otimes n}\otimes\v{f}^{\otimes m},\quad \v{Q}^{n,m}:=\v{q}^{\otimes m}\otimes\v{f}^{\otimes n}.
\end{equation}
Our main object of interest will be the \emph{conversion rate} $r_n:=m/n$, i.e., the number of target states one can obtain per one copy of the initial state. For notational clarity we will denote total initial and target distributions by $\v{P}^{n}$ and $\v{Q}^{n}$, with the dependence on $m$ (so, in fact, on $r_n$) kept implicit. The single-shot interconversion conditions can now be expressed as $\v{P}^n\prec\v{Q}^n$ for the entanglement and coherence transformations, and $\hat{\v{P}}^n\succ\hat{\v{Q}}^n$ for the thermodynamic transformations.

We also need to introduce the concept of approximate interconversion. Assume that for given $\v{P}^n$ and $\v{Q}^n$ the relevant majorisation relation does not hold, so that the interconversion is impossible. However, there may exist $\tilde{\v{Q}}^n$ that is $\epsilon$-close to $\v{Q}^n$ and such that the interconversion is possible. We then say that an approximate transformation is possible with the error level $\epsilon$ quantified by either the infidelity, $1-F$, or total variation distance (TVD),~$\delta$, between target and final states, with
\begin{subequations}
	\begin{align}
		F(\v{Q}^n,\tilde{\v{Q}}^n)&:=\left(\sum_i \sqrt{Q^n_i \tilde{Q}^n_i}\right)^2,\\
		\delta(\v{Q}^n,\tilde{\v{Q}}^n)&:=\frac{1}{2}\sum_i \left|Q^n_i-\tilde{Q}^n_i\right|.
	\end{align}
\end{subequations}
The concept of approximate interconversion gives rise to two notions of approximate majorisation introduced in Ref.~\cite{chubb2017beyond}, $\epsilon$-post-majorisation $\succ_\epsilon$ and $\epsilon$-pre-majorisation~$\prescript{}{\epsilon}{\succ}$, defined by
\begin{subequations}
	\begin{align}
		\v{a}\succ_\epsilon\v{b}&\quad\Longleftrightarrow\quad \exists\, \mathrlap{\tilde{\v{b}}}\phantom{\tilde{\v a}}:~\mathrlap{\v{a}}\phantom{\tilde{\v a}}\succ\tilde{\v{b}}~\mathrm{and}~\delta(\mathrlap{\v{b}}\phantom{\v a},\mathrlap{\tilde{\v{b}}}\phantom{\tilde{\v a}})\leq\epsilon,\\
		\v{a}\prescript{}{\epsilon}{\succ}~\v{b}&\quad\Longleftrightarrow\quad \exists\, \tilde{\v{a}}:~\tilde{\v{a}}\succ\mathrlap{\v{b}}\phantom{\tilde{\v b}}~\mathrm{and}~\delta(\v{a},\tilde{\v{a}})\leq\epsilon,
	\end{align}
\end{subequations}
where, depending on the context, $\delta$ can be replaced by $1-F$. Crucially, in Ref.~\cite{chubb2017beyond} the present authors showed that these two notions are equivalent and, moreover, that $\epsilon$-post-majorisation between embedded vectors, $\hat{\v{a}}\succ_\epsilon\hat{\v{b}}$, is a necessary and sufficient condition for the existence of an approximate thermodynamic transformation between $\v{a}$ and $\v{b}$ with error level $\epsilon$.

We conclude that an approximate transformation between initial and target states, represented by $\v{p}^{\otimes n}$ and $\v{q}^{\otimes nr_n}$, is possible within resource theories of entanglement and coherence if and only if
\begin{equation}
	\v{P}^n\prec_\epsilon\v{Q}^n,
\end{equation}
with the free state $\v{f}=\v{s}$. We will refer to the above relation as the approximate majorisation relation for the \emph{entanglement direction}. Similarly, such a transformation is possible within resource theory of thermodynamics if and only if
\begin{equation}
	\hat{\v{P}}^n\succ_\epsilon\hat{\v{Q}}^n,
\end{equation}
with the free state $\v{f}=\v{\gamma}$. We will refer to this relation as the approximate majorisation relation for the \emph{thermodynamic direction}.

\subsection{Information-theoretic notions}

The main role in the quantitative analysis of the interconversion process for the entanglement direction will be played by the Shannon entropy $H$ and entropy variance $V$. For a given probability distribution $\v{a}$ these are defined by
\begin{subequations}
	\begin{align}
		H(\v{a})&=-\sum_i a_i\ln a_i,\\
		V(\v{a})&=\sum_i a_i\left[\ln a_i+H(\v{a})\right]^2.
	\end{align}
\end{subequations}
The analogous role for the thermodynamic direction will be played by the relative entropy $D$ and relative entropy variance $V$. Given two probability distributions, $\v{a}$ and $\v{b}$, these are defined by
\begin{subequations}
	\begin{align}
		D(\v{a}||\v{b})&=\sum_i a_i\ln \frac{a_i}{b_i},\\
		V(\v{a}||\v{b})&=\sum_i a_i\left[\ln \frac{a_i}{b_i}-D(\v{a}||\v{b})\right]^2\!\!\!.
	\end{align}
\end{subequations}
An important fact, that can be verified by direct calculation, is that the relative quantities are invariant under embedding, i.e., \mbox{$D(\v{a}||\v{b})=D(\hat{\v{a}}||\hat{\v{b}})$} and \mbox{$V(\v{a}||\v{b})=V(\hat{\v{a}}||\hat{\v{b}})$}~\cite{chubb2017beyond}.

In order to formally state our main result we also need to introduce the notion of a \emph{moderate sequence}:

\begin{defn}[Moderate sequence]
	\label{def:moderate}
	A sequence of real numbers $\{t_n\}_n$ is a \emph{moderate sequence} if its scaling is strictly between $1/\sqrt{n}$ and $1$, meaning that $t_n\to 0$ and $\sqrt{n} t_n\to +\infty$ as $n\to\infty$.
\end{defn}

\noindent Note that an important family of moderate sequences is given by \mbox{$t_n\sim n^{-\alpha}$} for $\alpha\in(0,1/2)$, which can be used to obtain a particularly simple version of our main results. 

Finally, as we will be interested in asymptotic expansions in $n$, we will employ the standard asymptotic notation: $o(f(n))$, $O(f(n))$ and $\Theta(f(n))$. We will also use $\ev{>}$ and $\ev{<}$ to denote eventual inequalities, specifically we write $a_n\ev{>}b_n$ if and only if there exists $N$ such that $a_n>b_n$ for all $n\geq N$. Moreover, we will denote equalities and inequalities up to terms of order $o(t_n)$ by $\simeq$, $\lesssim$ and $\gtrsim$.

\section{Interconversion rates beyond the asymptotic regime}
\label{sec:result}

We are now ready to state our central technical result, which may be of interest outside the resource-theoretic studies due to ubiquity of majorisation partial order in the broad field of applied mathematics~\cite{marshall1979inequalities}. We split it into three theorems. The first two concern state interconversion below the asymptotic rate and with asymptotically vanishing error (one for each majorisation direction). The third one concerns practically less relevant scenario of state interconversion above the asymptotic rate and with error asymptotically approaching 1. 

For the entanglement direction we introduce the \emph{optimal conversion rate} $R^{\mathrm{ent}}_n(\epsilon)$ as the largest conversion rate $r_n$ for which the approximate majorisation relation for the entanglement direction, $\v{P}^n\prec_\epsilon\v{Q}^n$, holds. Due to the discussion presented in Sec.~\ref{sec:setting}, $R^{\mathrm{ent}}_n(\epsilon)$ is the maximal rate for which the approximate interconversion, with error $\epsilon$, is possible between states represented by $\v{p}$ and $\v{q}$ within resource theories of entanglement and coherence. We also define the \emph{asymptotic rate},
\begin{equation}
	R_\infty^\mathrm{ent}=\frac{H(\v{p})}{H(\v{q})},
\end{equation}
and the \emph{irreversibility parameter},
\begin{equation}
	\label{eq:nu_ent}
	\nu^\mathrm{ent}=\frac{V(\v{p})/H(\v{p})}{V(\v{q})/H(\v{q})}.
\end{equation}
We then have:

\begin{thm}[Entanglement direction]
	\label{thm:majorisation_ent}
	For any moderate sequence $t_n$ and the accepted error level of
	\begin{equation}
		\label{eq:error_ent}
		\epsilon_n= e^{-n t_n^2},
	\end{equation}
	the asymptotic expansion of the optimal conversion rate $R_n^{\mathrm{ent}}(\epsilon_n)$ is 		
	\begin{align}
		\label{eq:rate_ent}
		R_n^{\mathrm{ent}}(\epsilon_n)&\simeq R^{\mathrm{ent}}_\infty-\sqrt{\frac{2V(\v{p})}{H(\v{q})^2}}\left|1- 1/\sqrt{\nu^{\mathrm{ent}}}\right| t_n.
	\end{align}		
\end{thm}

Analogously, for the thermodynamic direction we introduce the \emph{optimal conversion rate} $R^{\mathrm{th}}_n(\epsilon)$ as the largest conversion rate $r_n$ for which the approximate majorisation relation for the thermodynamic direction, \mbox{$\hat{\v{P}}^n\succ_\epsilon \hat{\v{Q}}^n$}, holds. As before, $R^{\mathrm{th}}_n(\epsilon)$ is the maximal rate for which the approximate interconversion, with error $\epsilon$, is possible between states represented by $\v{p}$ and $\v{q}$ within the resource theory of thermodynamics. We also define the \emph{asymptotic rate},
\begin{equation}
	R_\infty^\mathrm{th}=\frac{D(\v{p}||\v{\gamma})}{D(\v{q}||\v{\gamma})},
\end{equation}
and the \emph{irreversibility parameter},
\begin{equation}
	\label{eq:nu_th}
	\nu^\mathrm{th}=\frac{V(\v{p}||\v{\gamma})/D(\v{p}||\v{\gamma})}{V(\v{q}||\v{\gamma})/D(\v{q}||\v{\gamma})}.
\end{equation}
We then have:

\begin{thm}[Thermodynamic direction]
	\label{thm:majorisation_th}
	For any moderate sequence $t_n$ and the accepted error level of
	\begin{equation}
		\label{eq:error_th}
		\epsilon_n= e^{-n t_n^2},
	\end{equation}
	the asymptotic expansion of the optimal conversion rate $R_n^{\mathrm{th}}(\epsilon_n)$ is	
	\begin{align}
		\label{eq:rate_th}
		R_n^{\mathrm{th}}(\epsilon_n)&\simeq R^{\mathrm{th}}_\infty-\sqrt{\frac{2V(\v{p}||\v{\gamma})}{D(\v{q}||\v{\gamma})^2}}\left|1- 1/\sqrt{\nu^{\mathrm{th}}}\right| t_n.
	\end{align}		
\end{thm}

Finally, one expects that conversion above the asymptotic rate leads to transformation error approaching 1. This is formalised in the following theorem which, unlike the previous two theorems (that hold for the error level measured by both infidelity and total variation distance), applies only to TVD. In Appendix~\ref{app:small}, where we relate our current results to the small deviation analysis of Refs.~\cite{kumagai2017second,chubb2017beyond}, we also conjecture the analogue of \cref{thm:converse} with the error measured by infidelity.

\begin{thm}[Converse regime]
	\label{thm:converse}
	For any moderate sequence $t_n$ and the accepted TVD error of
	\begin{equation}
		\label{eq:error_conv}
		\epsilon_n= 1-e^{-n t_n^2},
	\end{equation}
	the asymptotic expansion of the optimal conversion rate $R_n^{\mathrm{ent}}(\epsilon_n)$ is
	\begin{subequations}		
		\begin{align}
			\label{eq:rate_ent_conv}
			R_n^{\mathrm{ent}}(\epsilon_n)&\simeq R^{\mathrm{ent}}_\infty+\sqrt{\frac{2V(\v{p})}{H(\v{q})^2}}\left(1+ 1/\sqrt{\nu^{\mathrm{ent}}}\right)\,t_n,
		\end{align}	
	and similarly for $R_n^{\mathrm{th}}(\epsilon_n)$ we have
		\begin{align}
			\label{eq:rate_th_conv}
			R_n^{\mathrm{th}}(\epsilon_n)&\simeq R^{\mathrm{th}}_\infty+\sqrt{\frac{2V(\v{p}||\v{\gamma})}{D(\v{q}||\v{\gamma})^2}}\left(1+1/\sqrt{\nu^{\mathrm{th}}}\right)\,t_n.
		\end{align}	
	\end{subequations}	
\end{thm}

We present the proofs in Sec.~\ref{sec:proof}, after we introduce the necessary tools in Sec.~\ref{sec:tail_bounds}. Before that let us make two important remarks.
 
\begin{rmk}
	For initial and target states satisfying \mbox{$\nu^{\mathrm{ent}}=1$}, the optimal conversion rate $R_n^{\mathrm{ent}}$ in the regime of vanishing error is given by the asymptotic rate $R_\infty^{\mathrm{ent}}$. This means that, up to terms of order $o(t_n)$, such a transformation is reversible even for finite $n$. Analogous observation holds for the thermodynamic direction. We discuss the implications of this particularly interesting scenario in an accompanying paper~\cite{korzekwa2018avoiding}.	
\end{rmk}

\begin{rmk}
	When $V(\v{p})=0$, resulting in $1/\sqrt{\nu^{\mathrm{ent}}}$ diverging to infinity and the apparent multiplication of zero times infinity, one can simply use the definition of $\nu^{\mathrm{ent}}$ to replace Eq.~\eqref{eq:rate_ent} with
	\begin{align}
		R_n^{\mathrm{ent}}(\epsilon_n)&\simeq R^{\mathrm{ent}}_\infty\pm\sqrt{\frac{2V(\v{q})H(\v{p})}{H(\v{q})^3}} t_n.
			\end{align}			
	Analogous observation holds for the thermodynamic direction.
\end{rmk}

\section{Moderate deviation toolkit}
\label{sec:tail_bounds}

\subsection{Preliminaries}

The central result of the moderate deviation analysis can be stated as follows.

\begin{lem}[Moderate deviation bound]
	\label{lem:moderate}	
	Let $\{X_i\}_{1\leq  i\leq n}$ be independent and identically distributed (i.i.d.) random variables with zero-mean and variance $v$. For any moderate sequence $\{t_n\}_n$ the following hold:
	\begin{subequations}
		\begin{align}
			\label{eq:moderate_1}
			\lim_{n\rightarrow\infty} \frac{1}{nt_n^2}\ln\left[ \mathrm{Pr}\left(\frac{1}{n}\sum_{i=1}^n X_i\geq t_n \right)\right]=&-\frac{1}{2v},\\
			\label{eq:moderate_2}
			\lim_{n\rightarrow\infty} \frac{1}{nt_n^2}\ln\left[ \mathrm{Pr}\left(\frac{1}{n}\sum_{i=1}^n X_i\leq -t_n\right)\right]=&-\frac{1}{2v}.
		\end{align}
	\end{subequations}
\end{lem}

The proof of the above lemma can be found, e.g., in Appendix~A of Ref.~\cite{Chubb2017}. For the remainder of the paper, consider $\lbrace t_n\rbrace_n$ to be a fixed moderate sequence. For clarity we will henceforth omit the dependence of all implicit constants on this sequence. It should be noted that the above lemma also holds when $v=0$, where we henceforth adopt the convention that $1/v=+\infty$ in this case.

\subsection{Two variations on tail bounds}

We now want to adapt \cref{lem:moderate} to our purposes of majorisation-based analysis. For a probability vector $\v{a}$ we thus introduce the following quantity
\begin{align}
	k_n(\v{a},x)&:=\exp\left(H(\bm a^{\otimes n})+x n t_n\right),
\end{align}
which allows us to formulate the magnitude-based version of the moderate deviation bound for products of distributions.

\begin{lem}[Magnitude-based tail bound]
	\label{lem:magnitude_tail}
	Consider an arbitrary probability distribution $\bm a$. For $x\leq 0$ we have
	\begin{subequations}
		\begin{align}
			\label{eq:magnitude_1}
			\lim\limits_{n\to\infty} \frac{1}{nt_n^2} \ln 
			\sum_{i}
			\left\lbrace 
			\left(\bm a^{\otimes n}\right)_i		
			\middle|
			\left(\bm a^{\otimes n}\right)_i\geq \frac{1}{k_n(\v{a},x)}		
			\right\rbrace
			=\frac{-x^2}{2V(\v{a})}	,
		\end{align}	
		and similarly for $x\geq 0$ we have
		\begin{align}
			\label{eq:magnitude_2}
			\lim\limits_{n\to\infty} \frac{1}{nt_n^2} \ln 
			\sum_{i}
			\left\lbrace 
			\left(\bm a^{\otimes n}\right)_i		
			\middle|
			\left(\bm a^{\otimes n}\right)_i\leq \frac{1}{k_n(\v{a},x)}		
			\right\rbrace
			=\frac{-x^2}{2V(\v{a})}.
		\end{align}	
	\end{subequations}
\end{lem}

\begin{proof}
	Consider the random variable $L:=-\log a$, distributed according to $\v{a}$, such that the expectation value $\langle L\rangle$ and the variance \mbox{$\Var(L)$} are equal to $H(\bm a)$ and $V(\bm a)$ respectively. We can express $k_n$ in terms of $L$ as
	\begin{align}
		\log k_n(\v{a},x)=n\langle L\rangle_{\v{a}}+x n t_n.
	\end{align}
	If we let $\lbrace L_j\rbrace_{1\leq j\leq n}$ be i.i.d.\ copies of $L$, then we can write the tail bound of $\bm a^{\otimes n}$ in terms of tail bounds on the average of these variables, 
	\begin{align}
		&\sum_{i}
		\left\lbrace 
		\left(\bm a^{\otimes n}\right)_i 
		\middle|
		\left(\bm a^{\otimes n}\right)_i\geq \frac{1}{k_n(\v{a},x)}		
		\right\rbrace\notag\\
		&\quad=\sum_{i_1,\dots,i_n}\left\lbrace \prod_{j=1}^{n}a_{i_j} \middle|\prod_{j=1}^{n}a_{i_j}\geq \frac{1}{k_n(\v{a},x)}	\right\rbrace\notag\\
		&\quad=\sum_{i_1,\dots,i_n}\left\lbrace \prod_{j=1}^{n}a_{i_j} \middle|\sum_{j=1}^{n}\log a_{i_j}\geq \log \frac1{k_n(\v{a},x)}	\right\rbrace\notag\\
		&\quad=\Pr\left[\sum_{j=1}^n L_j\leq n\langle L\rangle+ x n t_n\right].
	\end{align}
	For $x<0$, we can now apply \cref{lem:magnitude_tail} to the variables $X_j:= \left(L_j-\langle L\rangle\right)/x$ to obtain Eq.~\eqref{eq:magnitude_1}. An analogous argument can be employed for $x>0$, with all of the above inequalities reversed, yielding Eq.~\eqref{eq:magnitude_2}. Finally, for \mbox{$x=0$} case, we can appeal to the Central Limit Theorem, which gives
	\begin{align}
	&\sum_{i}\left\lbrace 
	\left(\bm a^{\otimes n}\right)_i \Bigm|
	\left(\bm a^{\otimes n}\right)_i\geq \frac1{k_n(\v{a},0)}		
	\right\rbrace\notag\\
	&\qquad\qquad=\Pr\left[\frac{1}{n}\sum_{j=1}^{n}L_j\leq \langle L\rangle\right]\xrightarrow[]{n\to\infty} \frac{1}{2},
	\end{align}
	implying Eqs.~\eqref{eq:magnitude_1}-\eqref{eq:magnitude_2}.
\end{proof}

Using the above result we can now prove the majorisation-based version of the moderate deviation bound.
\begin{lem}[Majorisation-based tail bound]
	\label{lem:majorisation_tail}
	Consider an arbitrary probability distribution $\bm a$ satisfying $V(\bm a)>0$. For $x\leq 0$ we have
	\begin{subequations}
		\begin{align}
			\label{eq:majorisation_1}
			\lim\limits_{n\to\infty} \frac{1}{nt_n^2} \ln \left[\sum_{i\leq k_n(\v{a},x)} 
			(\bm a^{\otimes n})^\downarrow_i\right]	=-\frac{x^2}{2V(\v{a})},
		\end{align}
		and similarly for $x\geq 0$ we have
		\begin{align}
			\label{eq:majorisation_2}
			\lim\limits_{n\to\infty} \frac{1}{nt_n^2} \ln \left[\sum_{i\geq k_n(\v{a},x)}(\bm a^{\otimes n})_i^\downarrow\right]=-\frac{x^2}{2 V(\v{a})}.
		\end{align}
	\end{subequations}
\end{lem}

\begin{proof}
	Here we follow the proof of the small-deviation analogue of this result, Lemmas~15~and~16 of Ref.~\cite{kumagai2017second}. Consider first the $x\leq 0$ case, and define two sets of indices
	\begin{subequations}
		\begin{align}
			S_n(x)&:=\left\lbrace 1,\dots,\lfloor k_n(\bm a,x)\rfloor \right\rbrace,\\
			\tilde S_n(x)&:=\left\lbrace i \,\middle|\, (\bm a^{\otimes n})_i^\downarrow \geq 1/k_n(\bm a,x) \right\rbrace.
		\end{align}
	\end{subequations}
	We note that \cref{lem:magnitude_tail} gives that 
	\begin{align}
		\lim\limits_{n\to\infty}\frac{1}{nt_n^2}\ln\left[\sum_{i\in \tilde S(x)}\left(\bm a^{\otimes n}\right)_i^\downarrow\right]=-\frac{x^2}{2V(\bm a)}
	\end{align}
	for any $x\leq 0$, and we wish to show an analogous result for $ S_n(x)$. We will achieve this by showing, for any $\delta>0$, that $\tilde S_n(x)\subseteq S_n(x) \subseteq \tilde S_n(x+\delta)$ holds eventually, i.e., for large enough $n$. The first inclusion follows trivially from the normalisation of our distribution, and so it is left only to show that \mbox{$\tilde S_n(x) \subseteq S_n(x+\delta)$}. 
	
	Noting that $(\bm a^{\otimes n})_i^\downarrow-1/k_n(\bm a,x+\delta/2)\geq 0$ if and only if $i\in \tilde S_n(x+\delta/2)$,  we see that 
	\begin{align}
		&\sum_{i\in \tilde S_n(x+\delta/2)}\left[ (\bm a^{\otimes n})_i^\downarrow-\frac1{k_n(\bm a,x+\delta/2)} \right] \notag
		\\
		&\qquad\qquad\qquad
		\geq \sum_{i\in T}\left[ (\bm a^{\otimes n})_i^\downarrow-\frac1{k_n(\bm a,x+\delta/2)} \right],
	\end{align}
	for any set of indices $T$. Taking $T= \tilde S_n(x+\delta)$, this gives
	\begin{align}
		\label{eqn:cardinality}
		\frac
		{\abs{\tilde S_n(x+\delta)\setminus \tilde S_n(x+\delta/2)}}{k_n(\bm a,x+\delta/2)}\geq\!\!\!\!\!\! \sum_{i\in \tilde S_n(x+\delta)\setminus \tilde S_n(x+\delta/2)}\!\!\!\!\!\!\!\!\!\!\!\!\left(\bm a^{\otimes n}\right)_i^\downarrow.
	\end{align}	
	Lemma~\ref{lem:magnitude_tail} tells us that the summation on the RHS scales as $e^{-\Theta(nt_n^2)}$, specifically that there is a lower bound of the form $e^{-Cnt_n^2}$ for some constant $C$. Since $t_n\to 0$, we eventually have that $Ct_n<\delta/2$, and so this sum can be lower bounded as follows
	\begin{align}
		\sum_{i\in \tilde S_n(x+\delta)\setminus \tilde S_n(x+\delta/2)}\left(\bm a^{\otimes n}\right)_i^\downarrow \ev > e^{-\delta nt_n/2}.
	\end{align}
	Applying this bound to Eq.~\eqref{eqn:cardinality} allows us to conclude
	\begin{align}
	\abs{\tilde S_n(x+\delta)}
	&\geq \abs{\tilde S_n(x+\delta) \setminus  \tilde S_n(x+\delta/2)}\notag\\
	&\ev> e^{-\delta nt_n}k_n(\bm a,x+\delta/2),\notag\\
	&= k_n(\bm a,x),
	\end{align}
	and therefore that $S_n(x)\subseteq \tilde S_n(x+\delta)$ as required.
	
	The inclusions $\tilde S_n(x)\subseteq S_n(x) \subseteq \tilde S_n(x+\delta)$, together with \cref{lem:magnitude_tail}, give us the following inequalities
	\begin{subequations}
		\begin{align}
			\liminf\limits_{n\to\infty}\frac{1}{nt_n^2}\ln\left[\sum_{i\in S_n(x)}\left(\bm a^{\otimes n}\right)_i^\downarrow\right]
			&\geq -\frac{(x+\delta)^2}{2V(\bm a)},\\
			\limsup\limits_{n\to\infty}\frac{1}{nt_n^2}\ln\left[\sum_{i\in S_n(x)}\left(\bm a^{\otimes n}\right)_i^\downarrow\right]
			&\leq-\frac{x^2}{2V(\bm a)}.
		\end{align}
	\end{subequations}
 	As this holds for any $\delta>0$, we conclude that
	\begin{align}
		\lim\limits_{n\to\infty}\frac{1}{nt_n^2}\ln\left[\sum_{i\in S_n(x)}\left(\bm a^{\otimes n}\right)_i^\downarrow\right]=-\frac{x^2}{2V(\bm a)},
	\end{align} 
	which is equivalent to \cref{eq:majorisation_1}. An analogous proof can be performed for $x\geq 0$, resulting in Eq.~\eqref{eq:majorisation_2}. 
\end{proof}

\begin{rmk}
	One can extend Lemma~\ref{lem:majorisation_tail} to probability distributions $\v{a}$ with $V(\v{a})=0$ by a direct calculation, since $V(\v{a})=0$ means all non-zero entries of $\v{a}$ are equal. One then obtains that Eq.~\eqref{eq:majorisation_1} holds for $x<0$ and Eq.~\eqref{eq:majorisation_2} for $x>0$, i.e., both expressions diverge to $-\infty$.
\end{rmk}

\subsection{Tail bounds for total distributions}

Recall that in Sec.~\ref{subsec:multicopy} we defined total initial and target states for a given rate $r_n$ as
\begin{equation}
	\v{P}^n:=\v{p}^{\otimes n}\otimes\v{f}^{\otimes nr_n},\quad \v{Q}^n:=\v{q}^{\otimes nr_n}\otimes\v{f}^{\otimes n},
\end{equation}
where $\v{f}$ stands for the free state of a given resource theory, i.e., $\v{f}$ is a sharp state $\v{s}$ for entanglement and coherence transformations, and $\bm f$ is the maximally mixed state~$\bm \eta$ in the case of thermodynamic transformations (corresponding to the embedded thermal state $\v{\gamma}$). For notational clarity we will henceforth omit the $\downarrow$ superscripts on these total states, assuming them to be ordered (i.e. we denote $\bm P^{n\downarrow}$ and $\bm Q^{n\downarrow}$ simply by $\bm P^{n}$ and $\bm Q^{n}$).

Analogous to the quantity which appears in our moderate deviation bounds, consider the quantity
\begin{align}
	K_n(x)&:= \exp\left(H(\bm Q^n)+x n t_n \right).
\end{align}
Using Lemma~\ref{lem:majorisation_tail} we can prove the following tail bounds for the total distributions.

\begin{lem}[Tail bound for $\bm P^n$ and $\bm Q^n$]
	\label{lem:PQ_tail}
	For any $\mu\in\mathbb R$, consider the conversion rate
	\begin{align}
		\label{eq:rate_tail_bound}
		r_n(\mu)=\frac{H(\bm f)-H(\bm p)+\mu t_n}{H(\bm f)-H(\bm q)},
	\end{align}
	and the irreversibility parameter
	\begin{align}
		\nu := \frac{V(\bm p)}{V(\bm q)}\cdot \frac{H(\bm f)-H(\bm q)}{H(\bm f)-H(\bm p)}.
	\end{align}
	The total output state $\bm Q^n$ has the tail bounds
	\begin{subequations}
		\begin{align}
			x\leq 0:\quad \lim\limits_{n\to \infty}\frac{1}{nt_n^2}\ln\left[\sum_{i\leq K_n(x)}Q_i^{n} \right]&=-\frac{\nu x^2}{2V(\v{p})},\\
			x\geq 0:\quad \lim\limits_{n\to \infty}\frac{1}{nt_n^2}\ln\left[\sum_{i\geq K_n(x)}Q_i^{n} \right]&=-\frac{\nu x^2}{2V(\v{p})}.
		\end{align}
	\end{subequations}
	Similarly, the total input state $\bm P^n$ has the tail bounds
	\begin{subequations}
		\begin{align}
			x\leq\mu:~\lim\limits_{n\to \infty}\frac{1}{nt_n^2}\ln\left[\sum_{i\leq K_n(x)}P_i^{n} \right]&=-\frac{(x-\mu)^2}{2V(\v{p})},\\
			x\geq \mu:~\lim\limits_{n\to \infty}\frac{1}{nt_n^2}\ln\left[\sum_{i\geq K_n(x)}P_i^{n} \right]&=-\frac{(x-\mu)^2}{2V(\v{p})}.
		\end{align}
	\end{subequations}
\end{lem}

\begin{proof}
	For $\v{f}=\v{s}$ we have
	\begin{align}
		\sum_{i\leq K_n(x)}Q_i^{n}
		&=\sum_{i\leq k_{nr_n}(\bm q,x/r_n)}\left(\bm q^{\otimes nr_n}\otimes \bm s^{\otimes n}\right)_i^\downarrow\notag\\
		&=\sum_{i\leq k_{nr_n}(\bm q,x/r_n)}\left(\bm q^{\otimes nr_n}\right)_i^\downarrow.
	\end{align}

	Similarly for $\v{f}=\v{\eta}$ we have
	\begin{align}
		\sum_{i\leq K_n(x)}Q_i^{n}
		&=\sum_{i\leq d^nk_{nr_n}(\bm q,x/r_n)}\left(\bm q^{\otimes nr_n}\otimes \bm \eta^{\otimes n}\right)_i^\downarrow\notag\\
		&=\sum_{i\leq k_{nr_n}(\bm q,x/r_n)}\left(\bm q^{\otimes nr_n}\right)_i^\downarrow.
	\end{align}
	Applying \cref{lem:majorisation_tail} to both of the above equations yields the desired bounds.
	
	Next, define $K_n^{P}(y):= \exp\left(H(\bm P^n)+y n t_n \right)$. By analogy to $\bm Q^n$, we have the following tail bounds on $\bm P^n$
	\begin{subequations}
		\begin{align}
			\label{eq:KP_tail_1}
			y\leq 0:~\lim\limits_{n\to \infty}\frac{1}{nt_n^2}\ln\left[\sum_{i\leq K^P_n(y)}P_i^{n} \right]&=-\frac{y^2}{2V(\v{p})},\\
			\label{eq:KP_tail_2}
			y\geq 0:~\lim\limits_{n\to \infty}\frac{1}{nt_n^2}\ln\left[\sum_{i\geq K^P_n(y)}P_i^{n} \right]&=-\frac{y^2}{2V(\v{p})}.
		\end{align}
	\end{subequations}
	Using the rate $r_n(\mu)$, and expanding out both $K_n$ and $K_n^P$, we find that $K_n^P(x-\mu)=K_n(x)$. Substituting this into the above expressions, we get the desired tail bounds purely in terms of $K_n(x)$.
\end{proof}

We now want to consider the regions in which our two total distributions majorise each other. To do this, we first define the values of $x$ for which the tail bounds for $\bm P^{n}$ and $\bm Q^{n}$ coincide. Let us introduce
\begin{align}
	z_{\mathrm{C}}&:=\frac{\mu}{1-\sqrt \nu},\qquad\text{and}\qquad
	z_{\mathrm{T}}:=\frac{\mu}{1+\sqrt \nu}.
\end{align}
These correspond to the values of $x$ for which the moderate deviation tail bounds of the total distributions meet on the same side (cis) or on opposite sides (trans), respectively. More precisely, as a consequence of Lemma~\ref{lem:PQ_tail}, $z_{\mathrm{C}}$ and $z_{\mathrm{T}}$ are the solutions to the following equations
\begin{subequations}
	\begin{align}
		\!\!\!\lim\limits_{n\to \infty}\frac{1}{nt_n^2}\ln\!\left[\sum_{i\leq K_n(z_{\mathrm {C}})}\!\!\!\!\!\! P_i^n \right]&\!\!=
		\!\!\lim\limits_{n\to \infty}\frac{1}{nt_n^2}\ln\!\left[\sum_{i\leq K_n(z_{\mathrm {C}})}\!\!\!\!\!\! Q_i^n \right]\!\!,\\
		\!\!\!\lim\limits_{n\to \infty}\frac{1}{nt_n^2}\ln\!\left[\sum_{i\geq K_n(z_{\mathrm {C}})}\!\!\!\!\!\! P_i^n \right]&\!\!=
		\!\!\lim\limits_{n\to \infty}\frac{1}{nt_n^2}\ln\!\left[\sum_{i\geq K_n(z_{\mathrm {C}})}\!\!\!\!\!\! Q_i^n \right]\!\!,\\
		\!\!\!\lim\limits_{n\to \infty}\frac{1}{nt_n^2}\ln\!\left[\sum_{i\leq K_n(z_{\mathrm {T}})}\!\!\!\!\!\! P_i^n \right]&\!\!=
		\!\!\lim\limits_{n\to \infty}\frac{1}{nt_n^2}\ln\!\left[\sum_{i\geq K_n(z_{\mathrm {T}})}\!\!\!\!\!\! Q_i^n \right]\!\!,\\
		\!\!\!\lim\limits_{n\to \infty}\frac{1}{nt_n^2}\ln\!\left[\sum_{i\geq K_n(z_{\mathrm {T}})}\!\!\!\!\!\! P_i^n \right]&\!\!=
		\!\!\lim\limits_{n\to \infty}\frac{1}{nt_n^2}\ln\!\left[\sum_{i\leq K_n(z_{\mathrm {T}})}\!\!\!\!\!\! Q_i^n \right]\!\!.
	\end{align}
\end{subequations}
We schematically present the positions of $z_{\mathrm{C}}$ and $z_{\mathrm{T}}$ in Fig.~\ref{fig:cdfs}, which also serves to illustrate the proof of the following lemma.

\begin{lem}[Dominance of total states]
	\label{lem:dominance}
	For a bounded interval $[a,b]$, such that $a>z_{\mathrm C}$ (for~$\nu<1$) or $b<z_{\mathrm C}$ (for~$\nu>1$), the inequalities
	\begin{align}
		\label{ineq:dominance}
		\sum_{i \leq K_n(x)}P_i^n > \sum_{i \leq K_n(x)}Q_i^n 
	\end{align}
	hold for all $x\in[a,b]$, for sufficiently large $n$. Similarly, for any bounded interval $[a,b]$ with $b<z_{\mathrm T}$, the inequalities
	\begin{align}
		\label{ineq:dominance2}
		\sum_{i \geq K_n(x)}P_i^n > \sum_{i \leq K_n(x)}Q_i^n 
	\end{align}
	hold for all $x\in[a,b]$, for sufficiently large $n$.
\end{lem}

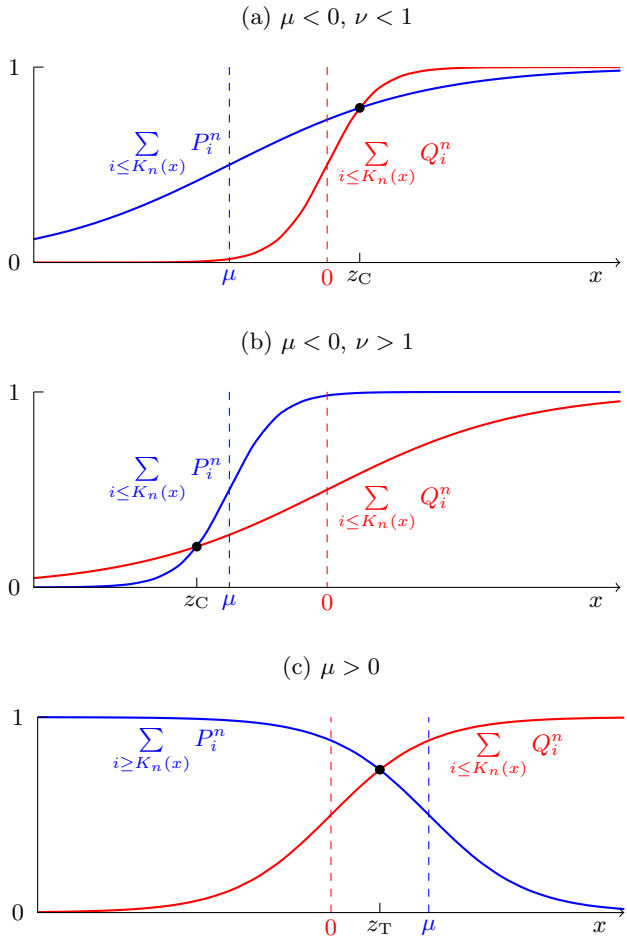
\begin{figure}
	\centering
	\begin{tikzpicture}[scale=1.3]
	\begin{scope}
	\clip (-3,-1) rectangle (3,1);
	
	\draw[domain=-3:3,smooth,variable=\x,red,thick] plot ({\x},{tanh(2*(\x))});
	\draw[domain=-3:3,smooth,variable=\x,blue,thick] plot ({\x},{tanh(0.5*(\x+1))});
	
	\draw[dashed,red] (0,-2) -- (0,2);
	\draw[dashed,blue] (-1,-2) -- (-1,2);
	
	\end{scope}
	\draw[->] (-3,1) -- (-3,-1) -- (3,-1);
	\draw (-3,1) -- (-2.9,1);
	\node at (-3.2,-1) {0};
	\node at (-3.2,+1) {1};
	
	\fill (1/3,0.58278294534791012006763998724863620140070458328815291696) circle (0.05cm);
	\node at (0,-1.15) {\color{red}0};
	\node at (-1,-1.15) {\color{blue}$\mu$};
	\draw (1/3,-1) -- (1/3,-.9);
	\node at (1/3,-1.15) {$z_{\mathrm C}$};
	
	\node at (2.75,-1.15) {$x$};
	
	\node at (.7,0) {\color{red}$\sum\limits_{i\leq K_n(x)}Q_i^n$};
	\node at (-1.65,.1) {\color{blue}$\sum\limits_{i\leq K_n(x)}P_i^n$};
	\node at (0,1.5) {(a) $\mu<0$, $\nu<1$};
	\end{tikzpicture}
	~\\
	~\\
	\begin{tikzpicture}[scale=1.3]
	\begin{scope}
	\clip (-3,-1) rectangle (3,1);
	
	\draw[domain=-3:3,smooth,variable=\x,red,thick] plot ({\x},{tanh(0.5*(\x))});
	\draw[domain=-3:3,smooth,variable=\x,blue,thick] plot ({\x},{tanh(2*(\x+1))});
	
	\draw[dashed,red] (0,-2) -- (0,2);
	\draw[dashed,blue] (-1,-2) -- (-1,2);
	
	\end{scope}
	\draw[->] (-3,1) -- (-3,-1) -- (3,-1);
	\draw (-3,1) -- (-2.9,1);
	\node at (-3.2,-1) {0};
	\node at (-3.2,+1) {1};
	
	\fill (-1-1/3,-0.58278294534791012006763998724863620140070458328815291696) circle (0.05cm);
	\node at (0,-1.15) {\color{red}0};
	\node at (-1,-1.15) {\color{blue}$\mu$};
	\draw (-1-1/3,-1) -- (-1-1/3,-.9);
	\node at (-1-1/3,-1.15) {$z_{\mathrm C}$};
	
	\node at (2.75,-1.15) {$x$};
	
	\node at (.7,-.2) {\color{red}$\sum\limits_{i\leq K_n(x)}Q_i^n$};
	\node at (-1.65,0.1) {\color{blue}$\sum\limits_{i\leq K_n(x)}P_i^n$};
	\node at (0,1.5) { (b) $\mu<0$, $\nu>1$};
	\end{tikzpicture}
	~\\
	~\\
	\begin{tikzpicture}[scale=1.3]
	\begin{scope}
	\clip (-3,-1) rectangle (3,1.2);
	
	\draw[domain=-3:3,smooth,variable=\x,red,thick] plot ({\x},{tanh(1*(\x))});
	\draw[domain=-3:3,smooth,variable=\x,blue,thick] plot ({\x},{tanh(1*(-\x+1))});
	
	\draw[dashed,red] (0,-1) -- (0,1);
	\draw[dashed,blue] (1,-1) -- (1,1);
	
	\end{scope}
	\draw[->] (-3,1) -- (-3,-1) -- (3,-1);
	\draw (-3,1) -- (-2.9,1);
	\node at (-3.2,-1) {0};
	\node at (-3.2,+1) {1};
	
	\fill (0.5,0.462117157260009758502318483643672548730289280330113038552) circle (0.05cm);
	\node at (0,-1.15) {\color{red}0};
	\node at (1,-1.15) {\color{blue}$\mu$};
	\draw (.5,-1) -- (.5,-.9);
	\node at (.5,-1.15) {$z_{\mathrm T}$};
	
	\node at (2.75,-1.15) {$x$};
	
	\node at (1.8,.6) {\color{red}$\sum\limits_{i\leq K_n(x)}Q_i^n$};
	\node at (-1.65,.65) {\color{blue}$\sum\limits_{i\geq K_n(x)}P_i^n$};
	\node at (0,1.5) { (c) $\mu>0$};
	\end{tikzpicture}
	
	\caption{Schematic representation of cumulative distribution functions for $\bm P^n$ and $\bm Q^n$, and the positions of $z_{\mathrm{C}}$ and $z_{\mathrm{T}}$ in different regimes.}
	\label{fig:cdfs}
\end{figure}

\begin{proof}
	We will prove Eq.~\eqref{ineq:dominance} and explain how Eq.~\eqref{ineq:dominance2} can be proven in an analogous way. Consider the function $L(y)=\log \frac{y}{1-y}$, which is strictly increasing for $y\in(0,1)$. Next, define two sequences of functions
	\begin{subequations}
		\begin{align}
			f_n(x)&:=\frac{2V(\bm p)}{nt_n^2}L\left(\sum_{i\leq K_n(x)}P^{n}_i\right),\\
			g_n(x)&:=\frac{2V(\bm p)}{nt_n^2}L\left(\sum_{i\leq K_n(x)}Q^{n}_i\right).
		\end{align}
	\end{subequations}
	We can combine the direct and converse parts of \cref{lem:majorisation_tail}, i.e., Eqs.~\eqref{eq:majorisation_1}~and~\eqref{eq:majorisation_2}, to obtain the limits $f_n\to f$ and $g_n\to g$ as $n\to \infty$, where
	\begin{subequations}
		\begin{align}
			f(x)=\frac{(x-\mu)^3}{\abs{x-\mu}}\quad\text{and}\quad g(x)=\frac{\nu x^3}{\abs{x}}.
		\end{align}
	\end{subequations}
	We now find that
	\begin{align}
		f(x)>g(x)
		\quad\Longleftrightarrow\quad
		\begin{dcases}
		x>z_{\mathrm C}&\text{for }\nu<1,\\
		x<z_{\mathrm C}&\text{for }\nu>1.\\
		\end{dcases}		
	\end{align}
	Therefore, for any $x$ in the above regions we have that \mbox{$f_n(x)\ev>g_n(x)$}. As $L$ is strictly monotone, this in turn implies that Eq.~\eqref{ineq:dominance} holds eventually in the same region. 
	
	The above argument only ensures that $f_n\to f$ and \mbox{$g_n\to g$} \emph{point-wise}. This does not yet allow us to conclude that there exists an $N$ such that Eq.~\eqref{ineq:dominance} will hold \emph{for all} $x\in[a,b]$ and all $n\geq N$. In \cref{app:uniformity} we close this gap by proving that this convergence is compact.
		
	Finally, Eq.~\eqref{ineq:dominance2} can be proven in an analogous way by substituting: $g_n\to -g_n$ and $g\to -g$.
\end{proof}

\section{Proof of the main theorem}
\label{sec:proof}

We are now ready to prove Theorems~\ref{thm:majorisation_ent}, \ref{thm:majorisation_th} and \ref{thm:converse}. We will achieve this in a series of steps. First, we will prove the following result.
\begin{prop}[Thermodynamic direction, $\v{\gamma}=\v{\eta}$, TVD]
	\label{prop:thermo}
	For a TVD error level of 
	\begin{align}
		\epsilon^-_n=e^{-nt_n^2}\quad\text{or}\quad\epsilon^+_n=1-e^{-nt_n^2},
	\end{align}
	the approximation majorisation condition 
	\begin{align}
		\label{eq:succ_transformation}
		\bm P^n\succ_{\epsilon^\pm_n} \bm Q^{n}
	\end{align}
		with $\bm f=\bm \eta$ holds with an optimal interconversion rate of
	\begin{align}
		\label{eq:succ_rate}
		R_n^{\mathrm{th}}(\epsilon_n^\pm)\simeq \frac{D(\bm p\|\bm \eta)\pm\sqrt{2V(\bm p\|\bm \eta)}\abs{1\pm1/\sqrt{\nu^{\mathrm{th}}}}t_n}{D(\bm q\|\bm \eta)}.
	\end{align}
\end{prop}
The proof of Proposition~\ref{prop:thermo} will consist of two parts: first, in Sec.~\ref{subsec:proof_achieve}, we will show that the claimed rate is achievable for the given error; and then, in Sec.~\ref{subsec:proof_opt}, that it is also optimal. This way we will prove a special case of Theorem~\ref{thm:majorisation_th} for the case of infinite temperature (when $\v{\gamma}=\v{\eta}$) and error level measured only by TVD; and Theorem~\ref{thm:converse} for the thermodynamic direction and infinite temperature. 

The next step is to generalise Proposition~\ref{prop:thermo} to arbitrary finite temperatures (arbitrary thermal state $\v{\gamma}$). It is enough to note that the approximate interconversion condition for the thermodynamic direction, $\hat{\v{P}}^n\succ_\epsilon\hat{\v{Q}}^n$, is exactly captured by Eq.~\eqref{eq:succ_transformation} if one only replaces $\v{p}$ and $\v{q}$ with $\hat{\v{p}}$ and $\hat{\v{q}}$, respectively. Moreover, since the relative entropy and relative entropy variance are invariant under the embedding map~\cite{chubb2017beyond}, one can obtain the optimal rate by ``unembeddinig'' Eq.~\eqref{eq:succ_rate}, i.e., replacing $\hat{\v{p}}$, $\hat{\v{q}}$ and $\v{\eta}=\hat{\v{\gamma}}$ with $\v{p}$, $\v{q}$ and~$\v{\gamma}$ respectively. Thus, by proving Proposition~\ref{prop:thermo}, we in fact prove Theorem~\ref{thm:majorisation_th} for any temperature and the error level measured by TVD; and Theorem~\ref{thm:converse} for the thermodynamic direction with arbitrary temperature. 

Then, in Sec.~\ref{subsec:proof_reversal}, we will prove the following result

\begin{prop}[Entanglement direction, TVD]
	\label{prop:ent}
	For a TVD error level of 
	\begin{align}
		\epsilon^-_n=e^{-nt_n^2}\quad\text{or}\quad\epsilon^+_n=1-e^{-nt_n^2},
	\end{align}
	the approximation majorisation condition 
	\begin{align}
		\label{eq:prec_transformation}
		\bm P^n\prec_{\epsilon^\pm_n} \bm Q^{n}
	\end{align}
	with $\bm f=\bm s$ holds with an optimal interconversion rate of
	\begin{align}
		\label{eq:prec_rate}
		R_n^{\mathrm{ent}}(\epsilon_n^\pm)\simeq \frac{H(\bm p)\pm\sqrt{2V(\bm p)}\abs{1\pm 1/\sqrt{\nu^{\mathrm{ent}}}}t_n}{H(\bm q)}.
	\end{align}
\end{prop}

To prove Proposition~\ref{prop:ent} we will leverage the proof of Proposition~\ref{prop:thermo}. More precisely, we will explain how to adapt that proof, so that the thermodynamic direction gets replaced by the entanglement direction. This way we will  prove Theorem~\ref{thm:majorisation_ent} for the error level measured by TVD; and Theorem~\ref{thm:converse} for the entanglement direction. 

The final missing piece is to show that Theorems~\ref{thm:majorisation_ent}~and~\ref{thm:majorisation_th} also hold for the error level measured by infidelity. We will prove this in Sec.~\ref{subsec:proof_fidelity}, again by explaining the necessary modifications of the reasoning that will result in replacing TVD with infidelity distance.

\subsection{Proof of \cref{prop:thermo} (Achieveability)}
\label{subsec:proof_achieve}

We will start by considering the \emph{achieveability} of \cref{prop:thermo}, i.e., a lower bound on the optimal conversion rate for the thermodynamic direction. For notational convenience, we will drop the superscripts on both $R_n^\mathrm{th}$ and $\nu^\mathrm{th}$, adopt the convention $D(\cdot):=D(\cdot\|\bm\eta)$, and note that $V(\bm \cdot \|\bm \eta) = V(\cdot)$. Specifically, we will prove the following:

\begin{lem}[\cref{prop:thermo}: Achieveability]
	\label{lem:achieve}
	For a TVD error level of
	\begin{align}
		\epsilon^-_n=e^{-nt_n^2}\quad\text{or}\quad\epsilon^+_n=1-e^{-nt_n^2},
	\end{align}
	the optimal rate is lower bounded,
	\begin{align}
		R_n(\epsilon_n^\pm)	\gtrsim\frac{D(\bm p)\pm\sqrt{2V(\bm p)}\abs{1\pm1/\sqrt{\nu}}t_n}{D(\bm q)}.
	\end{align}
\end{lem}

We will prove this lemma by constructing a family of distributions $\bmtilde P^{(\mu)}$, which eventually obey the required majorisation condition $\bmtilde P^{(\mu)}\succ \bm Q^n$. We will then show that by picking appropriate values of $\mu=\mu_+$ and $\mu=\mu_-$, one can obtain distributions $\bmtilde P^{(\mu_\pm)}$ such that \mbox{$\delta(\bm P^n,\bmtilde P^{(\mu_\pm)})\ev\leq \epsilon_n^\pm$}.

\subsubsection{Constructing the approximate distribution ${\mathbf {\tilde P}^{(\mu)}}$}

To prove achieveability, we will construct a family of distributions $\bmtilde P^{(\mu)}$ which, for any fixed $\mu$, eventually majorise $\bm Q^n$. As in \cref{lem:PQ_tail}, consider the rate
\begin{align}
	r_n(\mu)=\frac{D(\bm p)+\mu t_n}{D(\bm q)},
\end{align}
where $\mu\in \mathbb{R}$ is a parameter of our construction. We will construct $\bmtilde P^{(\mu)}$ by the \emph{cut-and-pile} method. Specifically we will consider starting with $\bm P^n$, removing mass from its tail, and adding it to the largest element. This construction allows us to construct a nearby state which is higher in the majorisation order. We start by defining the cutting point,
\begin{align}
	z_{\mu,\nu}
	:=\begin{dcases}
		2\mu-z_{\mathrm C} & :\mu<0,~~\nu <1,\\
		z_{\mathrm C} & :\mu<0,~~\nu >1,\\
		z_{\mathrm T} & :\mu >0.
	\end{dcases}
\end{align}
If we let $\zeta>0$ be a small slack parameter, then $\bmtilde P^{(\mu)}$ is defined as
\begin{align}
	\tilde P_i^{(\mu)}:= \begin{dcases}
		P^n_1+\sum_{i\geq K_n(y)} P_i^n & :i=1,\\
		P^n_i & :1<i<K_n(y),\\
		0 & :i\geq K_n(y),
	\end{dcases}
\end{align}
with $y=z_{\mu,\nu}-\zeta$.

\subsubsection{Showing majorisation ${\mathbf {\tilde P}}^{(\mu)}\succ {\mathbf{\tilde Q}}^n$}

Given the above construction, we now want to prove the majorisation condition $\bmtilde P^{(\mu)}\succ\bm Q^n$ eventually holds. The idea here is to leverage \cref{lem:dominance}, and show that a cut-and-pile construction with a cut at $K_n(x)$ for any $x<z_{\mu,\nu}$ will always eventually majorise $\bm Q^n$.

\begin{lem}
	\label{lem:achieve1}
	For any fixed $\mu$, $\bmtilde P^{(\mu)} \ev\succ \bm Q^n$.
\end{lem}

\begin{proof}
	To prove majorisation we need to show that, eventually, the inequalities
	\begin{align}
		\sum_{i=1}^{k}\tilde P_i^{(\mu)} \geq \sum_{i=1}^k Q^n_i 
		\label{ineq:majorisation}
	\end{align}
	hold for all $k$. The `cut' of the cut-and-pile construction implies Eq.~\eqref{ineq:majorisation} for large $k$. Specifically, the restricted support of $\bmtilde P^{(\mu)}$,
	\begin{align}
		\sum_{i< K_n(z_{\mu,\nu}-\zeta)}\tilde P_i^{(\mu)}=1,
	\end{align}
	implies that Eq.~\eqref{ineq:majorisation} holds trivially for any \mbox{$k\geq K_n(z_{\mu,\nu}-\zeta)$}. 
	
	The idea now is to show that the `pile' similarly gives us majorisation for small $k$, and then to leverage \cref{lem:dominance} to argue that $\bm P^n$ already majorises $\bm Q^n$ for intermediate $k$. We will split this argument into the three cases given in the definition of $z_{\mu,\nu}$, and illustrated by panels (a)-(c) of Fig.~\ref{fig:cdfs}.
	
	{\textbf{Case 1: }{$\mu<0$, $\nu<1$.}}
	Noticing that $\bm P^n$ tail bounds in \cref{lem:PQ_tail} are symmetric under \mbox{$x\to 2\mu-x$}, we have that
	\begin{align}
		&\lim\limits_{n\to\infty}\frac{1}{nt_n^2}\ln \left[\sum_{i\geq K_n(z_{\mu,\nu}-\zeta)}P^n_i\right]\notag\\
		&\qquad\qquad=
		\lim\limits_{n\to\infty}\frac{1}{nt_n^2}\ln \left[\sum_{i\leq K_n(z_{\mathrm C}+\zeta)}P^n_i\right].
	\end{align}
	Applying \cref{lem:dominance}, we therefore have that
	\begin{align}
		\sum_{i\geq K_n(z_{\mu,\nu}-\zeta)}P^{n}_i
		\ev >
		\sum_{i\leq K_n(z_{\mathrm C}+\zeta)}Q^{n}_i.
	\end{align}
	Using this, for any $k\leq K_n(z_{\mathrm C}+\zeta)$ we can leverage the `pile' in the $\bmtilde P^{(\mu)}$ construction to yield
	\begin{align}
		\sum_{i\leq k}\tilde P^{(\mu)}_i\geq\tilde P^{(\mu)}_1>\!\!\!\!\!\!\!\!\! 
		\sum_{i\geq K_n(z_{\mu,\nu}-\zeta)}\!\!\!\!\!\!\!\!\!\! P_i^n~ %\\
		\ev>\!\!\!\!\! \sum_{i\leq K_n(z_{\mathrm C}+\zeta)}\!\!\!\!\!\!\!\!\!\! Q_i^n%\\
		\geq  \sum_{i\leq k}Q_i^n,
	\end{align}
	which implies Eq.~\eqref{ineq:majorisation}. Applying \cref{lem:dominance}, we have that Eq.~\eqref{ineq:majorisation} must also eventually hold on the remaining intermediate indices \mbox{$k\in[K_n(z_{\mathrm C}+\zeta),K_n(z_{\mu,\nu}-\zeta)]$}.
	
	{\textbf{Case 2:} {$\mu<0$, $\nu>1$}}. Noticing that $\bm Q^n$ tail bounds in \cref{lem:PQ_tail} are symmetric under $x\to -x$, we have that
	\begin{align}
		&\lim\limits_{n\to\infty}\frac{1}{nt_n^2}\ln \left[\sum_{i\geq K_n(z_{\mu,\nu})}Q^n_i\right]\notag\\
		&\qquad\qquad=
		\lim\limits_{n\to\infty}\frac{1}{nt_n^2}\ln \left[\sum_{i\leq K_n(-z_{\mathrm C})}Q^n_i\right].
	\end{align}
	We now use \cref{lem:dominance} again, giving for any \mbox{$k\leq K_n(-z_{\mathrm C})$} that the `pile' of $\bmtilde P^{(\mu)}$ implies Eq.~\eqref{ineq:majorisation},
	\begin{align}
		\sum_{i\leq k}\tilde P^{(\mu)}_i 
		~>\!\!\!\! \sum_{i\geq K_n(z_{\mu,\nu}-\zeta)}\!\!\!\!\!\!\!\!\! P_i^n~ %\\
		\ev> \sum_{i\leq K_n(-z_{\mathrm C})}\!\!\!\!\!\!\!\!\! Q_i^n%\\
		\geq  \sum_{i\leq k}Q_i^n.
	\end{align}
	A direct application of Eq.~\eqref{lem:dominance} implies \cref{ineq:majorisation} also holds on the remaining intermediate indices \mbox{$k\in[K_n(-z_{\mathrm C}),K_n(z_{\mu,\nu}-\zeta)]$}.
	
	\textbf{Case 3:} {$\mu>0$}. For $k\leq  K_n(z_{\mathrm T}-\zeta)$, we use \cref{lem:dominance} to eventually give
	\begin{align}
		\sum_{i\leq k}\tilde P^{(\mu)}_i 
		>\!\!\! \sum_{i\geq K_n(z_{\mathrm T}-\zeta)}\!\!\!\!\!\! P_i^n %\\
		\ev>\!\!\! \sum_{i\leq K_n(z_{\mathrm T}-\zeta)}\!\!\!\!\!\! Q_i^n%\\
		\geq  \sum_{i\leq k}Q_i^n.
	\end{align}
	In this case there is no intermediate region, which implies Eq.~\eqref{ineq:majorisation} holds for all $k$.
\end{proof}

\subsubsection{Showing that $\mathbf{\tilde{P}}^{(\mu)}$ is close to $\mathbf{P}^n$}
\label{subsubsec:achieve_errors}

Now that we have shown that $\bmtilde P^{(\mu)}$ eventually majorises $\bm Q^n$, we want to ask how close $\bmtilde P^{(\mu)}$ is to $\bm P^n$ in total variation distance. The answer is provided by the following lemma.

\begin{lem}
	\label{lem:achieve2}
	For a fixed $\mu<0$,
	\begin{align}
		\delta(\bmtilde P^{(\mu)},\bm P^n) &\ev< \exp\left(-\frac{(z_{\mu,\nu}-\mu-2\zeta)^2}{2V(\v{p})}nt_n^2\right),
	\end{align}
	and similarly for $\mu>0$,
	\begin{align}
		\delta(\bmtilde P^{(\mu)},\bm P^n) &\ev< 1-\exp\left(-\frac{(z_{\mu,\nu}-\mu-2\zeta)^2}{2V(\v{p})}nt_n^2\right).
	\end{align}
\end{lem}

\begin{proof}
	A convenient feature of the cut-and-pile construction is that the total variation distance to the original distribution takes a particularly simple form. Specifically,
	\begin{equation}
		\delta(\bmtilde P^{(\mu)},\bm P^n)=\sum_{i\geq K_n(z_{\mu,\nu}-\zeta)}\!\!\!\!\!\!\!\!\!\! P^{n}_i.
	\end{equation}
	We can now apply \cref{lem:PQ_tail} to give
	\begin{subequations}
		\begin{align}
			\lim\limits_{n\to\infty} \frac{1}{nt_n^2}\ln \left[\delta(\bmtilde P^{(\mu<0)},\bm P^n)\right]
			&= -\frac{(z_{\mu,\nu}-\mu-\zeta)^2}{2V(\v{p})},\\
			\lim\limits_{n\to\infty} \frac{1}{nt_n^2}\ln \left[1-\delta(\bmtilde P^{(\mu>0)},\bm P^n)\right]
			&= -\frac{(z_{\mu,\nu}-\mu-\zeta)^2}{2V(\v{p})}.
		\end{align}
	\end{subequations}
	This, in turn, implies the eventual inequalities
	\begin{subequations}
		\begin{align}
			\ln \left[\delta(\bmtilde P^{(\mu<0)},\bm P^n)\right]
			&\ev < -\frac{(z_{\mu,\nu}-\mu-2\zeta)^2}{2V(\v{p})}nt_n^2,\\
			\ln \left[1-\delta(\bmtilde P^{(\mu>0)},\bm P^n)\right]
			&\ev > -\frac{(z_{\mu,\nu}-\mu-2\zeta)^2}{2V(\v{p})}nt_n^2,
		\end{align}
	\end{subequations}
	which are equivalent to our desired bounds.
\end{proof}

\subsubsection{Proof of Achieveability}

We now put together the above to prove the achieveability of \cref{prop:thermo}.

\begin{proof}[Proof of \cref{lem:achieve}]
	Consider our construction of $\bmtilde P^{(\mu)}$ for a specific choice of $\mu$. Specifically, let 
	\begin{align}
		\mu_\pm=\abs{1\pm1/\sqrt{\nu}}\left(\pm \sqrt{2V(\bm p)}-2\zeta\right),
	\end{align}
	where $\mu_-$ will give the direct result (with vanishing error $\epsilon_n^-$), and $\mu_+$ will give the converse result (with error $\epsilon_n^+$ approaching 1).
	
	From \cref{lem:achieve1} we have that $\bmtilde P^{(\mu_\pm)}\ev{\succ} \bm Q^n$ as required. Substituting $\mu_\pm$ in \cref{lem:achieve2}, we have that the TVD error is bounded,
	\begin{align}
		\delta\left(\bmtilde P^{(\mu_\pm)},\bm P^n\right)\ev<\epsilon^\pm_n,
	\end{align}
	as required. The rate in these cases takes the form
	\begin{align}
		\!\!\! r_n(\mu_\pm)
		&=\frac{D(\bm p)+\mu_\pm t_n}{D(\bm q)}\notag\!\\
		\!\!\!&=\frac{D(\bm p)+\abs{1\pm1/\sqrt{\nu}}\left(\pm \sqrt{2V(\bm p)}-2\zeta\right) t_n}{D(\bm q)}.\!
	\end{align}
	As the above analysis holds for any $\zeta>0$, we can therefore conclude that 
	\begin{align}
		R_n(\epsilon^\pm_n)	\gtrsim\frac{D(\bm p)\pm\sqrt{2V(\bm p)}\abs{1\pm1/\sqrt{\nu}}t_n}{D(\bm q)},
	\end{align}
	as desired.
\end{proof}

\subsection{Proof of \cref{prop:thermo} (Optimality)}
\label{subsec:proof_opt}

We now move on to showing the \emph{optimality} of \cref{prop:thermo}, i.e.\ an upper bound on the optimal conversion rate for the thermodynamic direction. For convenience we will reuse the notation used in the proof of achieveability.

\begin{lem}[\cref{prop:thermo}: Optimality]
	\label{lem:optimal}
	For a TVD error level of
	\begin{align}
		\epsilon^-_n=e^{-nt_n^2}\quad\text{or}\quad\epsilon^+_n=1-e^{-nt_n^2},
	\end{align}
	the optimal rate is upper bounded,
	\begin{align}
		R_n(\epsilon^\pm_n)	\lesssim\frac{D(\bm p)\pm\sqrt{2V(\bm p)}\abs{1\pm 1/ \sqrt{\nu}}t_n}{D(\bm q)}.
	\end{align}
\end{lem}

We will prove the above by showing that, for any distribution $\bmtilde P^n$ obeying the majorisation condition $\bmtilde P^n\succ \bm Q^n$ with a rate $r_n(\mu)$, the TVD distance between $\bmtilde P^n$ and $\bm P^n$ is eventually lower bounded by $\epsilon_n^\pm$. To achieve this, we will use the fact that the total variation distance is monotonically decreasing under coarse-graining. Specifically, for any distributions $\bm a$ and $\bm b$, and index $k$, the triangle inequality gives
\begin{align}
	\delta(\bm a,\bm b)&\geq \abs{\sum_{i\leq k}a_i-\sum_{j\leq k}b_j},
\end{align}
Applying this to distributions $\bmtilde P^n$ and $\bm P^n$, and index \mbox{$k=K_n(x)$}, gives
\begin{align}
	\delta(\bmtilde P^n,\bm P^n)&\geq 
	\abs{\sum_{i\leq K_n(x)}\tilde P_i^n-\sum_{j\leq K_n(x)}P^n_j}.
\end{align}
The idea is that, with a careful choice of $x$, we will be able to use the majorisation $\bmtilde P^n\succ \bm Q^n$ to replace the summations over $\bmtilde P^n$ with those over $\bm Q^n$. This will then allow us to apply \cref{lem:PQ_tail} to arrive at our final bound on the error. 

We will first present this argument in detail for the case where $\mu<0$ and $\nu <1$, and then present the modifications required for the remaining cases of $\mu<0$ and $\nu>1$, and $\mu>0$.

\subsubsection{Case 1: $\mu<0$, $\nu <1$}

Here we will perform our coarse-grained binning at \mbox{$x=z_{\mathrm C}-\zeta$}. Recalling that $\bmtilde P^n\succ \bm Q^n$, we have that
\begin{align}
	\sum_{i\leq K_n(z_{\mathrm C}- \zeta)}\tilde P^n_i\geq \sum_{i\leq K_n(z_{\mathrm C}-\zeta)} Q^n_i.
\end{align}
Using the positivity of $\zeta$, \cref{lem:dominance} allows us to also conclude that
\begin{align}
	\sum_{i\leq K_n(z_{\mathrm C}-\zeta)}\tilde P^n_i\ev> \sum_{i\leq K_n(z_{\mathrm C}-\zeta)} P^n_i.
\end{align}

We now use the fact that $\abs{\alpha-\beta}$ is monotonically increasing in $\alpha$ for $\alpha\geq \beta$. Using this, we have an eventual lower bound on TVD purely in terms of the total states $\bm P^n$ and $\bm Q^n$,
\begin{align}
	\delta(\bmtilde P^n,\bm P^n)&\geq \abs{\sum_{i\leq K_n(z_{\mathrm C}-\zeta)}Q_i^n-\sum_{j\leq K_n(z_{\mathrm C}- \zeta)}P^n_j}.
\end{align}

Applying \cref{lem:PQ_tail}, we see that the tail sum of $\bm Q^n$ asymptotically dominates. Specifically, we have
\begin{align}
	\liminf\limits_{n\to \infty} \frac{1}{nt_n^2}\ln\left[\delta(\bmtilde P^n,\bm P^n)\right]&\geq-\frac{\nu (z_{\mathrm C}-\zeta)^2}{2V(\bm p)},
\end{align}
and therefore
\begin{align}
	\delta(\bmtilde P^n,\bm P^n)&\ev> \exp\left(-\frac{\nu(z_{\mathrm C}- 2\zeta)^2}{2V(\bm p)}nt_n^2\right).
\end{align}
We now choose
\begin{align}
	\mu=
	\abs{1-1/\sqrt \nu }\left( -\sqrt{2V(\bm p)}+2\zeta \sqrt{\nu}\right).
\end{align}
This gives us that \mbox{$\delta(\bmtilde P^n,\bm P^n)\ev>\epsilon_n^-$}, with a rate of
\begin{align}
	\!\! r_n(\mu)\!=\!\frac{D(\bm p)+\abs{1\!-\! 1/\sqrt \nu }\left( -\sqrt{2V(\bm p)}+2\zeta \sqrt{\nu}\right)t_n}{D(\bm q)}.\!
\end{align}
As this is true for all $\zeta>0$, we can conclude that
\begin{align}
	R_n(\epsilon^-_n)	\lesssim\frac{D(\bm p)-\sqrt{2V(\bm p)}\abs{1- 1/\sqrt{\nu}}t_n}{D(\bm q)}.
\end{align}

\subsubsection{Case 2: $\mu<0$, $\nu >1$}

Now we consider the case of $\mu<0$ and $\nu>1$. The proof here is similar, starting with a cut at $x=z_{\mathrm C}+\zeta$. Here the tail sum of $\bm P^n$ asymptomatically dominates, with \cref{lem:PQ_tail} giving
\begin{align}
	\delta(\bmtilde P^n,\bm P^n)&\ev> \exp\left(-\frac{(z_{\mathrm C}-\mu+ 2\zeta)^2nt_n^2}{2V(\bm p)}\right).
\end{align}
We now make the choice
\begin{align}
	\mu=\abs{1-1/\sqrt{\nu}} \left(-\sqrt{2V(\bm p)}+2\zeta\right),
\end{align}
which gives that \mbox{$\delta (\bmtilde P^n,\bm P^n)\ev >\epsilon_n^-$}, and thus
\begin{align}
	R_n(\epsilon_n^-)\lesssim\frac{D(\bm p)-\sqrt{2V(\bm p)}\abs{1- 1/\sqrt{\nu}}t_n}{D(\bm q)}.
\end{align}

\subsubsection{Case 3: $\mu>0$}

Lastly, we consider the case of $\mu > 0$. Here, we perform our cut at $x=z_{\mathrm T}$ and, instead of using an argument based on majorisation \emph{from above}, we use majorisation \emph{from below}. Specifically, we note that $\bmtilde P^n\succ \bm Q^n$ implies
\begin{align}
	\sum_{i>K_n(z_{\mathrm T})}\tilde P^n_i\leq\sum_{i>K_n(z_{\mathrm T})}Q^n_i.
\end{align}
We now use that $\abs{\alpha-\beta}$ is monotonically \emph{decreasing} in $\alpha$ for $\alpha\leq \beta$. This gives us the analogous upper bound
\begin{align}
	\delta(\bmtilde P^n,\bm P^n)&\geq \abs{\sum_{i\leq K_n(z_{\mathrm T})}Q_i^n-\sum_{j\leq K_n(z_{\mathrm T})}P^n_j}.
\end{align}
Applying \cref{lem:PQ_tail}, we find that the tail sums of $\bm P^n$ and $\bm Q^n$ both dominate\footnote{The lack of a slack parameter in the cut means that these two sums compete. In the direct case (vanishing error $\epsilon^-_n$) this leads to a catastrophic cancellation, but in the converse case, this causes no problem.}, yielding
\begin{align}
	\delta(\bmtilde P^n,\bm P^n)&\ev> 1-\exp\left(-\frac{(z_{\mathrm T}-\mu+ \zeta)^2nt_n^2}{2V(\bm p)}\right).
\end{align}
for sufficiently large $n$. Substituting
\begin{align}
	\mu=\abs{1+1/\sqrt{\nu}}\left(\sqrt{2V(\bm p)}+\zeta\right),
\end{align}
we eventually have that \mbox{$\delta(\bmtilde P^n,\bm P^n)>\epsilon_n^+$}, and thus
\begin{align}
	R_n(\epsilon_n^+)\lesssim\frac{D(\bm p)+\sqrt{2V(\bm p)}\abs{1+1/\sqrt{\nu}}t_n}{D(\bm q)}.
\end{align}

\subsection{Proof of \cref{prop:ent}}
\label{subsec:proof_reversal}

\cref{prop:thermo} states that the largest rate $R_n$ such that the majorisation condition
\begin{align}
	\bm p^{\otimes n}\otimes \bm \eta^{\otimes nR_n}\succ_{\epsilon_n^\pm}\bm q^{\otimes nR_n}\otimes \bm \eta^{\otimes n}
\end{align}
holds is of the form
\begin{align}
	\!\!\! R_n(\epsilon_n^\pm)\simeq \frac{H(\bm \eta)-H(\bm p)\pm \sqrt{2V(\bm p)}\abs{1\pm1/\sqrt{\nu}}t_n}{H(\bm \eta)-H(\bm q)},
\end{align}	
where
\begin{align}
\nu = \frac{V(\bm p)}{V(\bm q)}\cdot \frac{H(\bm \eta)-H(\bm q)}{H(\bm \eta)-H(\bm p)}.
\end{align}
The proof of the above relied entirely on \cref{lem:PQ_tail}, which holds for $\bm f=\bm \eta$ as well as $\bm f=\bm s$. Thus, in order to rigorously prove \cref{prop:ent}, one could perform steps analogous to the ones presented in Sec.~\ref{subsec:proof_achieve}-\ref{subsec:proof_opt}, with $\v{\eta}$ replaced by $\v{s}$. Instead, below we present a shorter proof that directly employs the result stated by \cref{prop:thermo}.

\begin{proof}[Proof of \cref{prop:ent}]
	Since \cref{lem:PQ_tail} holds both for $\bm f=\bm \eta$ and for $\bm f=\bm s$, the statement of \cref{prop:thermo} can be extended to show that the smallest $R_n$ such that the majorisation condition
	\begin{align}
		\bm p^{\otimes n}\otimes \bm s^{\otimes nR_n}
		\succ_{\epsilon_n^\pm}
		\bm q^{\otimes nR_n}\otimes \bm s^{\otimes n}
	\end{align}
	holds is of the form
	\begin{align}
		\!\!\! R_n(\epsilon_n^\pm)\simeq \frac{H(\bm s)-H(\bm p)\pm \sqrt{2V(\bm p)}\abs{1\pm1/\sqrt{\nu}}t_n}{H(\bm s)-H(\bm q)},
	\end{align}	
	where
	\begin{align}
		\nu = \frac{V(\bm p)}{V(\bm q)}\cdot \frac{H(\bm s)-H(\bm q)}{H(\bm s)-H(\bm p)}.
	\end{align}
	
	We now want to reverse the direction of majorisation. The first step is simply to swap $\bm p\leftrightarrow \bm q$, and use the fact $\epsilon$-post-majorisation is equivalent to $\epsilon$-pre-majorisation~\cite{chubb2017beyond}. This transforms the considered majorisation relation into the following form
	\begin{align}
		\bm p^{\otimes nR_n}\otimes \bm s^{\otimes n}
		\prec_{\epsilon_n^\pm}
		\bm q^{\otimes n}\otimes \bm s^{\otimes nR_n}.
	\end{align}
	We see that $R_n$ now forms the `inverse rate' of the transformation between states represented by $\bm p$ and the ones represented by $\bm q$. To find the true rate of this transformation, we make the following substitutions
	\begin{align}
		n\leftarrow nR_n,\quad t_n\leftarrow t_n/\sqrt{R_n},\quad R_n\leftarrow 1/R_n.
	\end{align}
	As a result, the desired majorisation condition,
	\begin{align}
		\bm p^{\otimes n}\otimes \bm s^{\otimes nR_n}
		\prec_{\epsilon_n^\pm}
		\bm q^{\otimes nR_n}\otimes \bm s^{\otimes n}
	\end{align}
	holds with the optimal rate of the form 
	\begin{align}
		R_n(\epsilon_n^\pm)
		&\simeq \frac{H(\bm p)}{H(\bm q)\mp \sqrt{2V(\bm q)\frac{H(\bm q)}{H(\bm p)}}\abs{1\pm\sqrt{\nu}}t_n}\notag\\		
		&\simeq \frac{H(\bm p)\pm \sqrt{2V(\bm p)}\abs{1\pm 1/\sqrt{\nu}}}{H(\bm q)}
		\,,
	\end{align}
	and
	\begin{align}
		\nu= \frac{V(\bm p)}{V(\bm q)}\cdot \frac{H(\bm q)}{H(\bm p)}.
	\end{align}	
\end{proof}

\subsection{Extension to infidelity}
\label{subsec:proof_fidelity}

We now want to argue that, in the direct regime (for vanishing error $\epsilon_n^-$), our results extend to the case where we consider error in terms of infidelity instead of TVD. To show that the achieveability argument, presented in Sec.~\ref{subsec:proof_achieve}, extends to infidelity, we leverage the Fuchs-van de Graaf inequality~\cite{fuchs1999cryptographic}, specifically
\begin{align}
	1-\sqrt{F(\bm P^n,\bmtilde P^{n})}\leq \delta(\bm P^n,\bmtilde P^n).
\end{align}
Using this, in the direct regime where $\delta(\bm P^n,\bmtilde P^n)\to 0$, we have that the corresponding moderate exponential of infidelity must be bounded by that of the TVD,
\begin{align}
	&\limsup\limits_{n\to\infty}\frac1{nt_n^2}\ln\left[1-F(\bm P^n,\bmtilde P^n)\right]\notag\\
	&\qquad\qquad\qquad\leq\lim\limits_{n\to\infty}\frac1{nt_n^2}\ln\left[\delta(\bm P^n,\bmtilde P^n)\right].
\end{align}
Therefore, the TVD upper bounds employed in the proof of achieveability naturally extend to the infidelity.

In the proof of optimality, presented in Sec.~\ref{subsec:proof_opt}, we used coarse-graining and a monotonicity argument to bound the TVD as
\begin{align}
	\delta(\bm P^n,\bmtilde P^n)&\geq \abs{\sum_{i\leq k}P^n_i-\sum_{j\leq k}Q^n_j}.
\end{align}
By using the data-processing inequality and an analogous monotonicity argument, we can similarly bound the fidelity as
\begin{align}
	\sqrt{F(\bm P^n,\bmtilde P^n)}&\leq \sqrt{\sum_{i\leq k}P^n_i\cdot\sum_{j\leq k}Q^n_j}\notag\\
	&\qquad+\sqrt{\sum_{i> k}P^n_i\cdot\sum_{j> k}Q^n_j}.
\end{align}
In the case where $\nu<1$, we chose $k$ such that
\begin{align}
	\sum_{i\leq k}P^n_i
	\ll
	\sum_{i\leq k}Q^n_i
	\ll 1,
\end{align}
and similarly for $\nu>1$
\begin{align}
	\sum_{i>k}Q^n_i
	\ll
	\sum_{i>k}P^n_i
	\ll 1.
\end{align}
In either case, a single tail sum dominates in the bound upon both fidelity and TVD, similarly allowing us to lower bound the moderate exponent of the infidelity,
\begin{align}
	&\liminf\limits_{n\to\infty}\frac1{nt_n^2}\ln\left[1-F(\bm P^n,\bmtilde P^n)\right]\notag\\
	&\qquad\qquad\qquad\geq\lim\limits_{n\to\infty}\frac1{nt_n^2}\ln\left[\delta(\bm P^n,\bmtilde P^n)\right].
\end{align}

\section{Outlook}
\label{sec:outlook}

We have performed the moderate deviation analysis of resource interconversion problems for which single-shot transformation rules are based on majorisation and thermo-majorisation. As a result, in the regime of asymptotically vanishing error, we have found unified expressions for second-order corrections to asymptotic conversion rates within resource theories of entanglement, coherence and thermodynamics. More precisely, we obtained a family of results that specifies the optimal trade-off between the speed at which the conversion rate approaches the asymptotic rate, and the speed at which the error vanishes, when the number of transformed states $n$ grows. Crucially, we have found that the correction term can vanish independently of $n$ when a certain resonance condition between the initial and final states is satisfied. This opens the path to transformation reversibility beyond the asymptotic limit, the phenomenon that we discuss in detail in the accompanying paper~\cite{korzekwa2018avoiding}. 

There are quite a few research directions that one may want to take in order to generalise and extend the results presented in this paper. Since the small deviation analysis of the majorisation-based resource interconversion has been performed in Refs.~\cite{kumagai2017second,chubb2017beyond}, and the moderate deviation analysis was the focus of the current work, the straightforward extension would be to investigate the interconversion problem in the large deviation regime. This may be of particular interest in the context of fluctuation-free work extraction, where one may want to sacrifice a constant fraction of possible work output for its quality~\cite{aberg2013truly}. On the other hand, one could also look for the exact expression for the third-order term of the asymptotic expansion of $R_n$, which one can conjecture to scale as $O(t_n^2+\log n)$~\cite{Chubb2017}. Another obvious generalisation is to go beyond the restrictions of pure states (for entanglement direction) and energy-incoherent states (for thermodynamic direction). This, however, is a much harder task, as the single-shot transformation rules are still not known for these unrestricted cases. Finally, since the second-order analysis of resource interconversion led us not only to quantitative results, but also to qualitatively new predictions concerning finite-size reversibility, a similar analysis for other resource theories is now very well justified.

\subsection*{Acknowledgements} 

CC and KK acknowledge support from the ARC via the Centre of Excellence in Engineered Quantum Systems (EQuS), project number CE110001013. CC also acknowledges support from the Australian Institute for Nanoscale Science and Technology Postgraduate Scholarship (John Makepeace Bennett Gift). MT is funded by an ARC Discovery Early Career Researcher Award, project number DE160100821.

\appendix

\section{Relation to small deviation bound}
\label{app:small}

In this section we explore the relationship between our moderate deviation results of Theorems \ref{thm:majorisation_ent}--\ref{thm:converse}, and the small deviation results of Refs.~\cite{kumagai2017second,chubb2017beyond}. In these papers, a second-order expansion of the rate is given for a \emph{constant} infidelity error, in terms of the Rayleigh-normal distribution. Below we will consider the expansions of the Rayleigh-normal cumulative distribution function around $-\infty$, and show consistency with Theorems~\ref{thm:majorisation_ent}~and~\ref{thm:majorisation_th}. Moreover, we will explain how the expansion around $+\infty$ leads to a conjecture analogous to \cref{thm:converse} with error level measured by infidelity.

The Rayleigh-normal distributions are a parameterised family of distributions, depending on a parameter \mbox{$\nu\geq 0$}. For a formal definition and properties, see Ref.~\cite{kumagai2017second}. We will denote the associated cumulative distribution functions by $Z_\nu(\mu)$, for $\nu\geq 0$ and $\mu\in\mathbb R$. We will also adopt the notation for Gaussian probability and cumulative distribution functions of
\begin{align}
	\phi(x)&=\frac{1}{\sqrt{2\pi}}e^{-x^2/2},\\
	\phi_{\mu,\nu}(x)&=\frac{1}{\sqrt{2\pi\nu}}e^{-(x-\mu)^2/2\nu},\\
	\Phi(x)&=\int_{-\infty}^{x}\phi(t)\,\mathrm dt,\\
	\Phi_{\mu,\nu}(x)&=\int_{-\infty}^{x}\phi_{\mu,\nu}(t)\,\mathrm dt.
\end{align}

\subsection{Expansion around $\mu=-\infty$}

The crossing point $\alpha_{\mu,\nu}$ is defined by the equation~\cite{kumagai2017second}
\begin{align}
	\frac{\phi(\alpha_{\mu,\nu})}{\phi_{\mu,\nu}(\alpha_{\mu,\nu})}
	=\frac{\Phi(\alpha_{\mu,\nu})}{\Phi_{\mu,\nu}(\alpha_{\mu,\nu})}.
\end{align}
As $\mu\to -\infty$, $\alpha_{\mu,\nu}\to+\infty$. We will use the $x\to\infty$ approximation $\Phi(x)\approx 1$, which leads to \mbox{$\phi(\alpha_{\mu,\nu})/\phi_{\mu,\nu}(\alpha_{\mu,\nu})\approx 1$}, resulting in
\begin{align}
	\alpha_{\mu,\nu}\approx \frac{\mu}{1-\sqrt \nu}.
\end{align}

We now look at the Rayleigh-normal distribution, which takes the form
\begin{align}
	\sqrt{1-Z_\nu(\mu)}
	&=
	\sqrt{\frac{2\sqrt \nu}{1+\nu}} e^{-\mu^2/4(1+\nu)}\Phi_{\frac{-\mu}{1+\nu},\frac{2\nu}{1+\nu}}(-\alpha_{\mu,\nu})\notag\\
	&~~~~+\sqrt{\Phi(\alpha_{\mu,\nu})\Phi_{\mu,\nu}(\alpha_{\mu,\nu})}
	.
\end{align}
We now note that, as $\alpha_{\mu,\nu}\to +\infty$, the first term on the RHS is exponentially vanishing, and second exponentially approaching 1. Specifically we have
\begin{subequations}
\begin{align}
	\!\!\!\ln\left[1-\sqrt{\Phi(\alpha_{\mu,\nu})\Phi_{\mu,\nu}(\alpha_{\mu,\nu})}\right] &\approx -\frac{1}{2}\frac{\mu^2}{(1-\sqrt \nu)^2},\!\\
	\!\!\!\ln\left[e^{-\mu^2/4(1+\nu)}\Phi_{\frac{-\mu}{1+\nu},\frac{2\nu}{1+\nu}}(\alpha_{\mu,\nu})\right]
	&\approx -\frac{1}{2}\frac{\mu^2}{(1-\sqrt \nu)^2},\!
\end{align}
\end{subequations}
and thus for $\mu\to -\infty$ we have
\begin{align}
	\ln\left[Z_{\nu}(\mu)\right]\approx -\frac{1}{2}\frac{\mu^2}{(1-\sqrt \nu)^2}.
\end{align}

\subsection{Consistency with Theorems \ref{thm:majorisation_ent} and \ref{thm:majorisation_th}}

The small deviation analyses of Refs.~\cite{kumagai2017second,chubb2017beyond} consider transformations with a fixed infidelity, bounded away from zero. Despite this, using the above expansion one can na\"ively substitute $1-F=\epsilon_n^-$ into the expressions for the optimal rates obtained within the small deviation regime. Whilst this analysis is no longer rigorous, we will see that it gives a rate which is consistent with our rigorous results, Theorems~\ref{thm:majorisation_ent}~and~\ref{thm:majorisation_th}.

Inverting the expansion around $\mu=-\infty$, we have
\begin{align}
	Z^{-1}_\nu(\epsilon) \approx \abs{1-\sqrt{\nu}}\sqrt{-2\ln\epsilon},
\end{align}
for small positive $\epsilon$. In particular, for $\epsilon_n^{-}:=e^{-nt_n^2}$ we have
\begin{align}
	Z^{-1}_\nu(\epsilon_n^-) \simeq \abs{1-\sqrt{\nu}}\sqrt{2n}t_n.
\end{align}
Substituting the above into the results of Refs.~\cite{kumagai2017second,chubb2017beyond} yields the expressions for optimal rates from Theorems~\ref{thm:majorisation_ent}~and~\ref{thm:majorisation_th}.

\subsection{Expansion around $\mu= + \infty$. }

The crossing point $\alpha_{\mu,\nu}$ is defined by the equation~\cite{kumagai2017second}
\begin{align}
	\frac{\phi(\alpha_{\mu,\nu})}{\phi_{\mu,\nu}(\alpha_{\mu,\nu})}
	=\frac{\Phi(\alpha_{\mu,\nu})}{\Phi_{\mu,\nu}(\alpha_{\mu,\nu})}.
\end{align}
As $\mu\to +\infty$, $\alpha_{\mu,\nu}\to-\infty$. We can now use the approximation $\Phi(x)\approx\phi(x)/x$ for $x\to-\infty$. Applying this to the above, we have
\begin{align}
	\frac{\phi(\alpha_{\mu,\nu})}{\phi_{\mu,\nu}(\alpha_{\mu,\nu})}
	\approx\frac{\phi(\alpha_{\mu,\nu})}{\phi_{\mu,\nu}(\alpha_{\mu,\nu})}
	\cdot\frac{\alpha_{\mu,\nu}-\mu}{\nu\alpha_{\mu,\nu}},
\end{align}
which in turn implies that
\begin{align}
	\alpha_{\mu,\nu}\approx \frac{\mu}{1-\nu}.
\end{align}

Returning to the Rayleigh-normal distribution, we note that $\sqrt{1-Z_\nu(\mu)}$ exponentially vanishes as \mbox{$\alpha_{\mu,\nu}\to-\infty$}, specifically
\begin{subequations}
\begin{align}
	\!\!\!\!\!\ln\left[\sqrt{\Phi(\alpha_{\mu,\nu})\Phi_{\mu,\nu}(\alpha_{\mu,\nu})}\right] &\approx -\frac{1+\nu}{(1-\nu)^2}\frac{\mu^2}{4},\\
	\!\!\ln\left[e^{-\mu^2/4(1+\nu)}\Phi_{\frac{-\mu}{1+\nu},\frac{2\nu}{1+\nu}}(\alpha_{\mu,\nu})\right]&\approx -\frac{1}{1+\nu}\frac{\mu^2}{4}.\!
\end{align}
\end{subequations}
As such, we can see that the second expression dominates, leading us to conclude that, for $\mu\to+\infty$, the Rayleigh-normal distribution may be approximated as
\begin{align}
	\ln\left[1-Z_{\nu}(\mu)\right]\approx -\frac{\mu^2}{4(1+\nu)}.
\end{align}

\subsection{Conjectured infidelity analogue of \cref{thm:converse}}

Similar to the case of the expansion around $\mu=-\infty$, we can use the expansion around $\mu=+\infty$ to (non-rigorously) obtain a conjecture for the analogue of \cref{thm:converse} for error measured in infidelity. Inverting the expansion around $\mu=+\infty$, we find
\begin{align}
	Z_\nu^{-1}(\epsilon) \approx \sqrt{-4(1+\nu)\ln(1-\epsilon)},
\end{align}
for $\epsilon$ close to 1. Specifically, for $\epsilon_n^+=1-e^{-nt_n^2}$ we have
\begin{align}
	Z_\nu^{-1}(\epsilon_n^+) \approx \sqrt{4(1+\nu)n}t_n.
\end{align}
Substituting this approximation into the results of Refs.~\cite{kumagai2017second,chubb2017beyond}, we have that the optimal rates for an infidelity error level of $\epsilon_n^+$ are of the following form
\begin{subequations}		
	\begin{align}
	R_n^{\mathrm{ent}}(\epsilon_n)&\simeq R^{\mathrm{ent}}_\infty+\sqrt{\frac{2V(\v{p})}{H(\v{q})^2}}\sqrt{1+1/{\nu^{\mathrm{th}}}}\,t_n,
	\\
	R_n^{\mathrm{th}}(\epsilon_n)&\simeq R^{\mathrm{th}}_\infty+\sqrt{\frac{2V(\v{p}||\v{\gamma})}{D(\v{q}||\v{\gamma})^2}}\sqrt{1+1/{\nu^{\mathrm{th}}}}\,t_n.
	\end{align}	
\end{subequations}

\section{Compact convergence lemma}
\label{app:uniformity}

\begin{lem}
	\label{lem:eventually greater}
	Let $\lbrace f_n\rbrace_n$ and $\lbrace g_n\rbrace_n$ be sequences of non-decreasing real functions, both of which point-wise converge to continuous functions $f$ and $g$, respectively. If $f>g$, then $f_n|_X>g_n|_X$ eventually on all compact $X$.
\end{lem}
\begin{proof}
	First we note (see, e.g., page 1 of Ref.~\cite{resnick2013extreme})
	that sequences of non-decreasing real functions which point-wise converge to continuous functions do so compactly. Let $\Delta_n:X\to\mathbb{R}$ be defined by $\Delta_n(x)=f_n(x)-g_n(x)$ for all $x\in X$. As $\Delta=f-g$, we know that $\Delta>0$. Indeed, because $\Delta$ is a continuous function on a compact domain, the Extreme Value Theorem tells us that $\Delta$ is bounded away from zero, i.e., there exists an $\epsilon>0$ such that $\Delta\geq \epsilon$. As $\Delta_n\to \Delta$ uniformly, we must eventually have that $\Delta_n\geq \Delta-\epsilon/2\geq \epsilon/2$, and so $f_n|_X> g_n|_X$.
\end{proof}

\end{document}